\documentclass[conference]{IEEEtran}

\usepackage[mathscr]{euscript}
\usepackage{xspace}
\usepackage{graphicx}
\usepackage{algorithm}
\usepackage{algorithmicx}
\usepackage{algpseudocode}
\usepackage{amssymb}
\usepackage{amsmath}
\usepackage{amsthm}
\usepackage{array}
\usepackage{url}
\usepackage{pifont}
\usepackage{xcolor}
\usepackage{tablefootnote}
\usepackage{paralist}
\usepackage{subcaption}
\usepackage{booktabs} 
\usepackage{stmaryrd}
\usepackage{harpoon}
\usepackage{accents}
\usepackage{bbm}
\usepackage{times}

\newcommand{\method}{\textsc{SliceNDice}\xspace}
\newcommand{\seedMethod}{\textsc{GreedySeed}\xspace}

\newcommand{\methodT}{{\Large\textsc{SliceNDice}}\xspace}

\newcommand{\mass}{\texttt{Mass}\xspace}
\newcommand{\avgdeg}{\texttt{AvgDeg}\xspace}
\newcommand{\dens}{\texttt{Dens}\xspace}
\newcommand{\singval}{\texttt{SingVal}\xspace}
\newcommand{\cssusp}{\texttt{CSSusp}\xspace}

\definecolor{darky}{RGB}{178, 154, 57
}
\definecolor{darkg}{RGB}{52, 158, 85}

\newcommand{\shuffle}{\textsc{Shuffle}\xspace}
\newcommand{\choice}{\textsc{Choose}\xspace}
\newcommand{\seednodes}{\textsc{SeedNodes}\xspace}
\newcommand{\seedviews}{\textsc{SeedViews}\xspace}
\newcommand{\updatenodes}{\textsc{UpdateNodes}\xspace}
\newcommand{\updateviews}{\textsc{UpdateViews}\xspace}

\algdef{SE}[DOWHILE]{Do}{doWhile}{\algorithmicdo}[1]{\algorithmicwhile\ #1}%
\algnewcommand{\algorithmicgoto}{\textbf{go to}}%
\algnewcommand{\Goto}[1]{\algorithmicgoto~\ref{#1}}%

\newcommand{\mvere}{\texttt{MVERE}\xspace}

\newcommand{\tick}{\textcolor{darkg}{\ding{52}}}
\newcommand{\cross}{\textcolor{red}{\ding{56}}}

\newtheorem{problem}{Problem}
\newtheorem{lemma}{Lemma}

\newtheorem{axiom}{Axiom}
\newtheorem{defn}{Definition}

\newtheorem*{axiomnonumber}{Axiom}

\graphicspath{{./fig/}}

\newcommand{\set}[1]{
 \mathcal{#1}
}

\newcommand{\fun}[1]{
 \mathscr{#1}
}

\def\Rho{P}

\usepackage{centernot}
\usepackage{hyperref}
\usepackage{amsmath}

\def\BibTeX{{\rm B\kern-.05em{\sc i\kern-.025em b}\kern-.08em
    T\kern-.1667em\lower.7ex\hbox{E}\kern-.125emX}}
\begin{document}

\title{\method: Mining Suspicious Multi-attribute Entity Groups with Multi-view Graphs}
\author{
\IEEEauthorblockN{Hamed Nilforoshan \thanks{hiiiii}}
\IEEEauthorblockA{\textit{Columbia University} \\
New York City, NY \\
hn2284@columbia.edu}
\and
\IEEEauthorblockN{Neil Shah}
\IEEEauthorblockA{\textit{Snap Inc.} \\
Santa Monica, CA \\
nshah@snap.com}
}
\maketitle

\begin{abstract}
Given the reach of web platforms, bad actors have considerable incentives to manipulate and defraud users at the expense of platform integrity.  This has spurred research in numerous suspicious behavior detection tasks, including detection of sybil accounts, false information, and payment scams/fraud.  In this paper, we draw the insight that many such initiatives can be tackled in a common framework by posing a detection task which seeks to find groups of entities which share too many properties with one another across multiple attributes (sybil accounts created at the same time and location, propaganda spreaders broadcasting articles with the same rhetoric and with similar reshares, etc.) 
Our work makes four core contributions: Firstly, we posit a novel \emph{formulation} of this task as a multi-view graph mining problem, in which distinct views reflect distinct attribute similarities across entities, and contextual similarity and attribute importance are respected.  Secondly, we propose a novel \emph{suspiciousness metric} for scoring entity groups given the abnormality of their synchronicity across multiple views, which obeys intuitive desiderata that existing metrics do not.  Finally, we propose the \method \emph{algorithm} which enables efficient extraction of highly suspicious entity groups, and demonstrate its \emph{practicality} in production, in terms of strong detection performance and discoveries on Snapchat's large advertiser ecosystem (89\% precision and numerous discoveries of real fraud rings), marked  outperformance of baselines (over 97\% precision/recall in simulated settings) and linear scalability.

\end{abstract}

\begin{IEEEkeywords}
anomaly detection, attributed data, multi-view graphs, outlier
\end{IEEEkeywords}

\section{Introduction}
\begin{figure*}[t!]
\centering
%
\subcaptionbox{Our formulation: multi-attributed entity data as a multi-view graph\label{fig:crown_a}}
{
\includegraphics[width=0.64\textwidth]{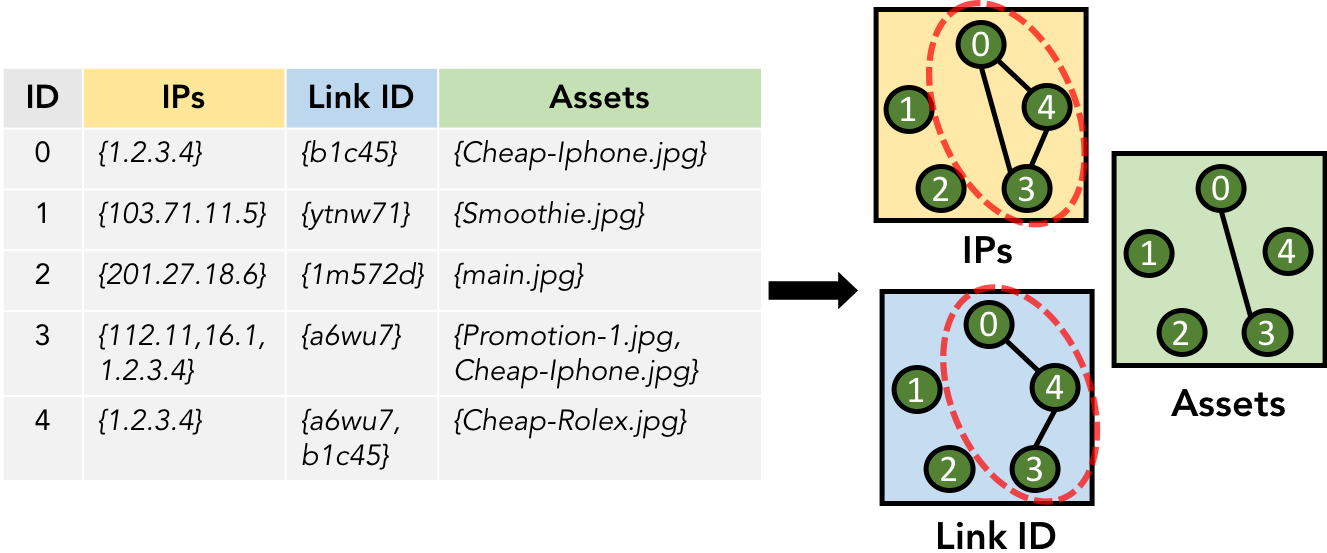}
}
\subcaptionbox{\method detected e-commerce fraud \label{fig:crown_b}}
{ 
\includegraphics[width=0.32\textwidth]{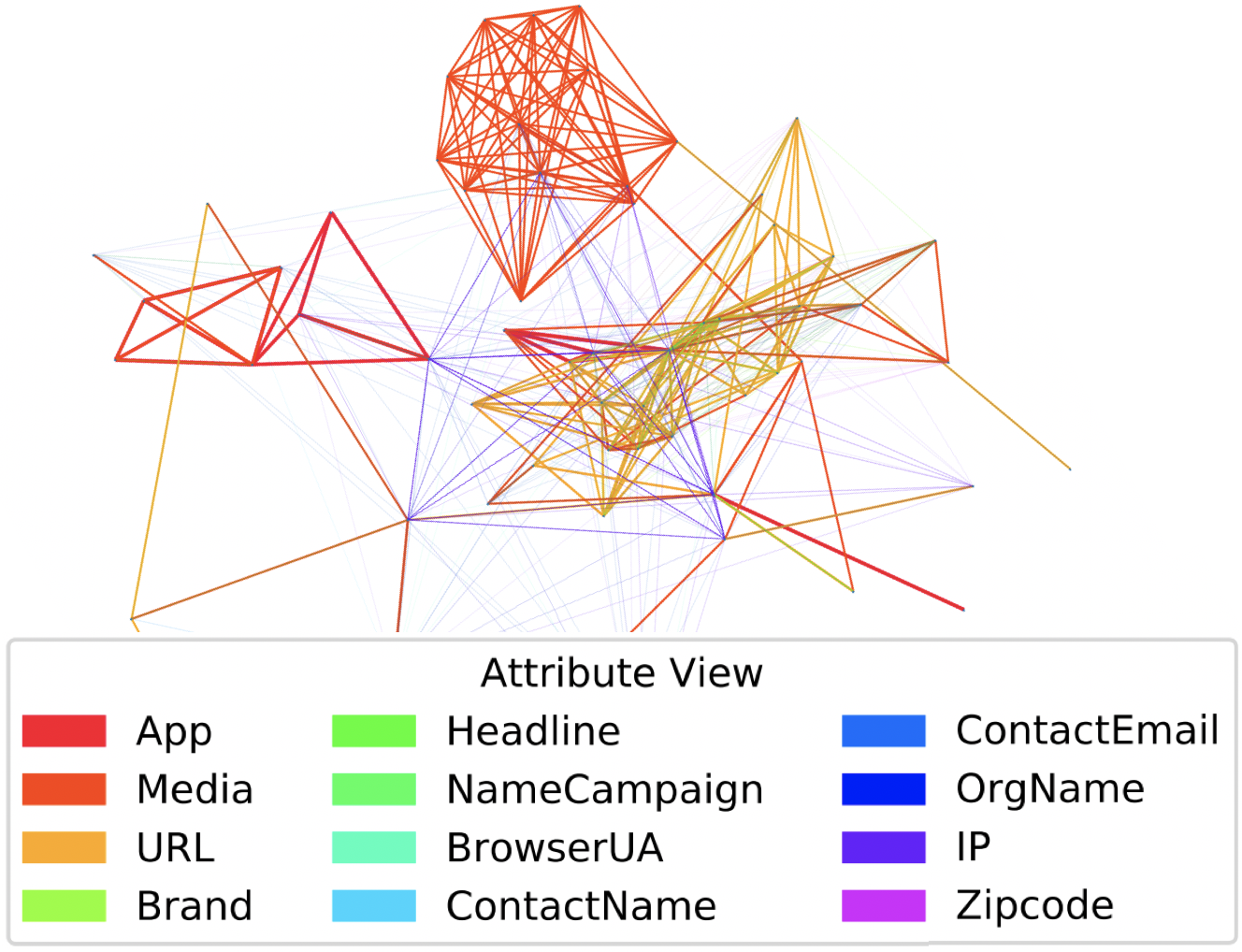}
}
\label{fig:ard}
\caption{Our work tackles suspicious behavior detection in multi-attribute entity data using a multi-view graph mining formulation, shown in (a).  Our \method algorithm offers scalable ranking and discovery of such behaviors across multiple entities and graph views, enabling us to uncover integrity violations like e-commerce fraud on the Snapchat ads platform, shown in (b).}
\label{fig:crown}
\end{figure*}

Online services and social networks (Snapchat, Quora, Amazon, etc.) have become common means of  engagement in human activities including socialization, information sharing and commerce. However, given their dissemination power and reach, they have been long-exploited by bad actors looking to manipulate user perception, spread misinformation, and falsely promote bad content. Such malicious activities are motivated by numerous factors including monetary incentives \cite{thomas2013trafficking} to personal \cite{smith2008cyberbullying} and political interests \cite{bessi2016social}.  
Tackling online misbehavior is a challenging research problem, given its variance with respect to manifestation across platforms, incentive structures and time.  Furthermore, online misbehavior is often adversarial in nature, and not solvable by systems which aim to cooperate with users to achieve positive outcomes~\cite{nilforoshan2018leveraging}.  Despite these challenges, researchers have made progress in a variety of scenarios including fake account creation \cite{cao2014uncovering, xiao2015detecting}, bot follower detection \cite{shah2014spotting, jiang2014catchsync}, malware detection \cite{chau2011polonium} and more.  However, many of these solutions require large, labeled datasets, and offer little to no extensibility.  Unfortunately, given the ever-increasing multitude of new abuse vectors, application-specific anti-abuse solutions and rich labeled sets are not always feasible or timely, motivating research towards flexible, unsupervised methods.

In this work, we propose an unsupervised solution for a wide-class of misbehavior problems in which we are tasked with discovering group-level suspicious behavior across attributed data. Our work is inspired by prior works, which have demonstrated that such behaviors often occur in lockstep and are better discernible on a group, rather than individual level \cite{kumar2018false, shah2017many}.   However, we tackle a problem which is left unaddressed by prior work: 
\begin{problem}[Informal]
\label{prob:inf}
How can we quantify and automatically uncover highly-suspicious entity groups with multiple associated attribute types and values?
\end{problem}
This setting naturally extends to a number of abuse detection scenarios, such as discovering online e-commerce scammers given profile creation, posting and e-mail address attributes, or pinpointing fraudulent advertisers given ad URLs, targeting criteria and key phrases. Our work leverages the insight that groups of entities who share too many, and too unlikely, attribute values are unusual and worth investigation.  We build on this intuition by (a) casting the problem of mining suspicious entity groups over multiple attributes as mining cohesive subgraphs across multiple views, (b) proposing a novel metric to quantify group suspiciousness in multi-view settings, (c) developing a scalable algorithm to mine such groups quickly, and (d) demonstrating effectiveness.  Specifically, our contributions are as follows.

\textbf{Formulation.} \emph{We propose modeling multi-attribute entity data as a multi-view graph, which captures notions of similarity via shared attribute values between entities}. In our model, each node represents an entity (e.g. account, organization), each view represents an attribute (e.g. files uploaded, IP addresses used, etc.) and each non-zero edge indicates some attribute value overlap (see Figure \ref{fig:crown_a}) 

\textbf{Suspiciousness Metric.} \emph{We design a novel suspiciousness metric, which quantifies subgraph (entity group) suspiciousness in multi-view graphs (multi-attribute settings)}. We identify desiderata that suspiciousness metrics in this setting should obey, and prove that our metric adheres to these properties while previously proposed options do not.  

\textbf{Algorithm.} \emph{We propose \method, an algorithm which scalably extracts highly suspicious subgraphs in multi-view graphs}.  Our algorithm uses a greedy, locally optimal search strategy to expand seeds of similar nodes into larger subgraphs with more cohesion.  We discuss design decisions which improve performance including careful seeding, context-aware similarity weighting and performance optimizations.  Our practical implementation leverages all these, enabling efficient entity group mining. 

\textbf{Practicality.} \emph{We evaluate \method both on real-world data and simulated settings, demonstrating high detection performance and efficiency.} We show effectiveness in detecting suspicious behaviors on the  Snapchat advertiser platform, with $12$ distinct attributes and over 230K entities (advertiser organizations).  Figure \ref{fig:crown_b} shows a large network of e-commerce fraudsters uncovered by our approach (nodes are advertiser organizations, and edge colors indicate shared attributes, ranked in suspiciousness).  We also conduct synthetic experiments to evaluate empirical performance against baselines, demonstrating significant improvements in  multi-attribute suspicious entity group detection.  Overall, our methods and results offer significant promise in bettering web platform integrity.
\section{Related Work}
\label{sec:relwork}
We discuss prior work in two related contexts below. 

\textbf{Mining entity groups.}
Prior works have shown that suspicious behaviors often manifest in synchronous group-level behaviors \cite{shah2017many, kumar2018false}.  
Several works assume inputs in the form of a single graph snapshot.  \cite{prakash2010eigenspokes, jiang2016inferring} mine communities using eigenplots from singular value decomposition (SVD) over adjacency matrices. \cite{charikar2000greedy,hooi2016fraudar,blondel2008fast} propose greedy pruning/expansion algorithms for identifying dense subgraphs.  \cite{chakrabarti2004fully, dhillon2003information} co-cluster nodes based on information theoretic measures relating to minimum-description length.  Some prior works \cite{atzmueller2016description,atzmueller2011efficient} tackle subgroup mining from a single graph which also has node attributes, via efficient subgroup enumeration using tree-based branch-and-bound approaches which rely on specialized community goodness scoring functions.  \emph{Unlike these works, our work entails mining suspicious groups based on overlap across multiple attributes and graph views, such that attribute importance is respected and an appropriate suspiciousness measure is used.}

Several works consider group detection in a multi-view (also known as multi-layer and sometimes multi-plex) formulation.  \cite{kim2015community} gives a detailed overview of community mining algorithms in this setting.  \emph{Our work is unlike these in that it is the only one which operates in settings with $>2$ views, considers different view importances, and allows for flexible view selection in subgroup discovery}.  Moreover, our work focuses on detection of \emph{suspicious} groups in a multi-view setting, rather than communities.  \cite{mao2014malspot} uses PARAFAC tensor decomposition to find network attacks. \cite{jiang2016spotting, beutel2013copycatch, shin2016mzoom} use greedy expansion methods for finding entity groups. \cite{shah2015timecrunch, metwally2015scalable} use single graph clustering algorithms, followed by stitching afterwards.  \emph{Our work differs in that we (a) cast multi-attribute group detection into a multi-view graph formulation, (b) simultaneously mine a flexible set of nodes and views to maximize a novel suspiciousness metric, and (c) use a compressed data representation, unlike other methods which incur massive space complexity from dense matrix/tensor representation.}

\textbf{Quantifying suspiciousness.} Most approaches quantify behavioral suspiciousness as likelihood under a given model.  Given labeled data, supervised learners have shown use in suspiciousness ranking tasks for video views \cite{chen2015analysis}, account names \cite{freeman2013using}, registrations \cite{xiao2015detecting} and URL spamicity \cite{ma2009beyond}.  However, given considerations of label sparsity, many works posit unsupervised and semi-supervised graph-based techniques.  \cite{shah2016edgecentric, hooi2016birdnest} propose information theoretic and Bayesian scoring functions for node suspiciousness in edge-attributed networks.  \cite{shah2014spotting, akoglu2010oddball} use reconstruction error from low-rank models to quantify suspicious connectivity.  \cite{lamba2017zoo} ranks node suspiciousness based on participation in highly-dense communities, 
\cite{gyongyi2004combating, guacho2018semi} exploit propagation-based based methods with few known labels to measure account sybil likelihood, review authenticity and article misinformation propensity.  \emph{Our work (a) focuses on groups rather than entities, and (b) needs no labels.} 

Several methods exist for group subgraph/tensor scoring; \cite{lee2010survey} overviews.  The most common are density \cite{lee2010survey}, average degree \cite{charikar2000greedy, hooi2016fraudar}, subgraph weight, singular value \cite{prakash2010eigenspokes}.  \cite{jiang2016spotting} defines suspiciousness using log-likelihood of subtensor mass assuming a Poisson model.  \emph{Our work differs in that we (a) quantify suspiciousness in multi-view graphs, and (b) show that alternative metrics are unsuitable in this setting.}

\section{Problem Formulation}
\label{sec:probform}

\begin{table}[t!]
\scriptsize
\centering
\begin{tabular}{ll}
\toprule
{\bf Symbol} & {\bf Definition} \\ \midrule
$K$ & Number of total attributes (graph views) \\
$ \set{A}_i$ & Set of possible values for attribute $i$ \\
$ \fun{A}_i(\cdot)$ & Maps nodes to attr. $i$ valueset: $\fun{A}_i: \set{G} \rightarrow 2^{\set{A}_i}$ \\
\midrule
$ \set{G}$ & Multi-view graph over all views \\ 
$ \vec{K}$ & Indicator for multi-view graph views, $\vec{K} \in \{0,1\}^{K}$ \\
$ \set{G}_i$ & Single ($i^{th}$) view of multi-view graph $\set{G}$ \\ 
$ \set{G}_{\vec{K}} $ & $\set{G}$ over chosen $\sum\vec{K}$ views, $\set{G}_{\vec{K}} = \{ \set{G}_i | \mathbbm{1}(K_i) \}$ \\
$ N $ & Number of nodes in $\set{G}$ \\
$ V $ & Volume of graph $\set{G}_i$: $N(N-1)/2$ \\
$ C_i $ & Mass of graph $\set{G}_i: \sum_{(a,b) \in \set{G}^2} \, w_i^{(a,b)}$ \\
$ \Rho_i $ & Density of graph $\set{G}_i$: ${C_i}/{V}$ \\
$ w_i^{(a,b)}$ & Edge weight between nodes $a, b$ in $\set{G}_i$, s.t. $w_i^{(a,b)} \in \mathbb{R}^+$ \\
\midrule

$ \set{X}$ & Multi-view subgraph over all views \\
$ \vec{k}$ & Indicator for multi-view subgraph views, $\vec{k} \in \{0,1\}^{K}$ \\
$ \set{X}_i$ & $i^{th}$ view of  multi-view subgraph $\set{X}$ \\
$ \set{X}_{\vec{k}} $ & $\set{X}$ over chosen $\sum\vec{k}$ views, $\set{X}_{\vec{k}} = \{ \set{X}_i | \mathbbm{1}(k_i) \}$ \\
$ n $ & Number of nodes in $\set{X}$ \\
$ v $ & Volume of subgraph of $\set{X}_i$: $n(n-1)/2$ \\
$ c_i $ & Mass of subgraph $\set{X}_i: \sum_{(a,b) \in \set{X}^2} \, w_i^{(a,b)}$ \\
$ \rho_i $ & Density of subgraph $\set{X}_i$: ${c_i}/{v}$ \\
$z$ & Constraint on chosen subgraph views, $|\vec{k}| = z$ \\
$ f(\cdot) $ & Mass-parameterized suspiciousness metric \\
$ \hat{f}(\cdot) $ & Density-parameterized suspiciousness metric \\
\bottomrule
\end{tabular}
\caption{\label{tbl:symb} Frequently used symbols and definitions.}
\end{table}
\setlength{\textfloatsep}{4pt}

The problem setting we consider (Problem \ref{prob:inf}) is commonly encountered in many real anti-abuse scenarios: a practitioner is given data spanning a large number of entities (i.e. users, organizations, objects) with a number of associated attribute types such as webpage URL, observed IP address, creation date, and associated (possibly multiple) values such as \emph{xxx.com}, \emph{5.10.15.20, 5.10.15.25}, \emph{01/01/19} and is tasked with finding suspicious behaviors such as fake sybil accounts or engagement boosters.  

In tackling this problem, one must make several considerations: \emph{What qualifies as suspicious? How can we quantify suspiciousness? How can we discover such behavior automatically?}  We build from the intuition that suspicious behaviors often occur in lockstep across \emph{multiple} entities (see Section \ref{sec:relwork}), and are most discernible when considering relationships \emph{between} entities.  For example, it may be challenging without context to determine suspiciousness of an advertiser linking to URL \emph{xxx.com}, logging in from IPs \emph{5.10.15.20} and with creation date \emph{01/01/19}; however, knowing that 99 other advertisers share these exact properties, our perception of suspiciousness increases drastically.  Thus, we focus on mining suspicious entity groups, where suspiciousness is governed by the degree of synchronicity between entities, and across various attributes.  

We draw inspiration from graph mining literature; graphs offers significant promise in characterizing and analyzing between-entity similarities, and are natural data structures for the same.  Graphs model individual entities as \emph{nodes} and relationships between them as \emph{edges}; for example, a graph could be used to describe \emph{who-purchased-what} relationships between users and products on Amazon.  
Below, we discuss our approach for tackling Problem \ref{prob:inf}, by leveraging the concept of \emph{multi-view} graphs.  We motivate and discuss three associated building blocks: (a) using multi-view graphs to model our multi-attributed entity setting, (b) quantifying group suspiciousness, and (c) mining highly suspicious groups in such multi-view graphs.  Throughout this paper, we utilize formal notation for reader clarity and convenience wherever possible; Table \ref{tbl:symb} summarizes the frequently used symbols and definitions, partitioned into attribute-related, graph-related and subgraph-related notation. 

\subsection{Representing Multi-attribute Data}
In this work, we represent multi-attribute entity data as a \emph{multi-view graph} (MVG).  An MVG is a type of graph in which we have multiple \emph{views} of interactions, usually in the form of distinct edge types; to extend our previous example, we could consider who-purchased-what, who-rates-what and who-viewed-what relationships as different views between users and products.  Each view of an MVG can individually be considered as a single facet or mode of behavior, and spans over the same, fixed set of nodes.  In our particular setting, we consider a representation of multi-attribute data in which we have a fixed set of nodes (entities), and their relationships across multiple views (attributes).  Thus, edges in each graph view represent relationships between entities, and are weighted based on their attribute similarities in that attribute space.

Formally, we have a set of $N$ entities with $K$ associated attribute types over the attribute value spaces $\set{A}_1 \ldots \set{A}_K$.  For notational convenience, we introduce attribute-value accessor mapping functions $\fun{A}_1 \ldots \fun{A}_K$ for the $K$ attribute spaces respectively, such that $\fun{A}_i: \set{G} \rightarrow 2^{\set{A}_i}$.  Effectively, $\fun{A}_i(a)$ denotes the subset of attribute values from $\set{A}$ associated with node $a \in \set{G}$.  We can construe this as an MVG $\set{G}$ on $N$ nodes (entities) and $K$ views (attributes), such that $\set{G}$ is a set of individual graph views $\{\set{G}_1 \ldots \set{G}_K\}$.  For convenience, we introduce notations $\set{G}_i$ and $\set{G}_{\vec{K}}$, to refer to a specific graph view, and a specific subset of graph views ($\vec{K}$ is a $K$-length vector indicator) respectively. We consider an edge $a \leftrightarrow b$ with $w_i^{(a,b)} > 0$ in view $\set{G}_i$ to exist between nodes $a,b$ if $\fun{A}_i(a) \cap \fun{A}_i(b) \neq \varnothing$, or informally $a$ and $b$ share at least one common feature value on the $i^{th}$ attribute.  If $\fun{A}_i(a) \cap \fun{A}_i(b) = \varnothing$ (no overlap between feature values), we consider that no edge between $a,b$ exists in $\set{G}_i$, or equivalently that $w_i^{(a,b)} = 0$.  In general, we consider non-negative weights, so that $w_i^{(a,b)} \in \mathbb{R}^+$.  We can consider many weighting strategies, but posit the notion that large weight between $\fun{A}_i(a)$ and $\fun{A}_i(b)$ indicates intuitively higher, or more rare similarities.

\subsection{Quantifying Group Suspiciousness}
\label{sec:proposed_axioms}

\begin{figure}
    \centering
    \includegraphics[width=\linewidth]{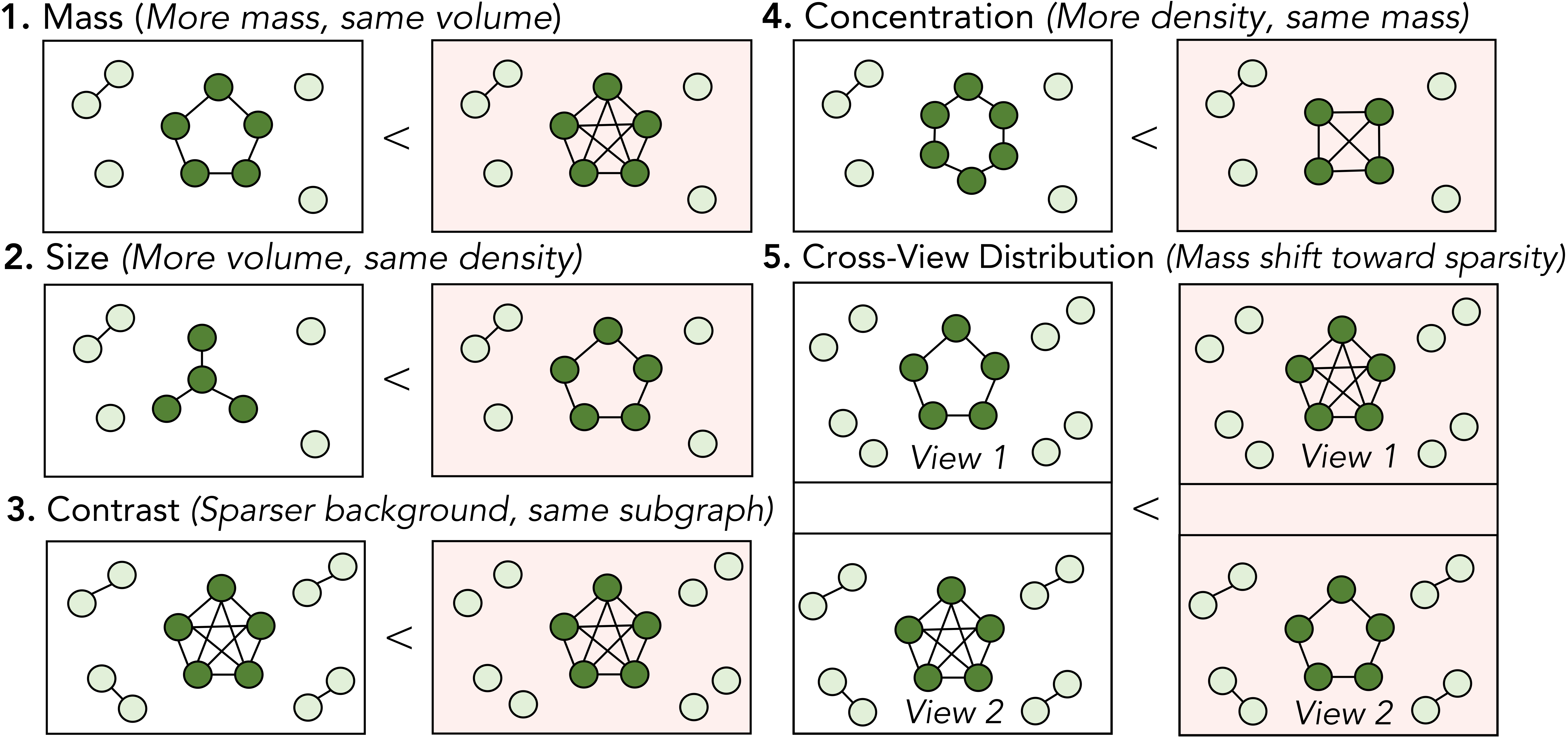}
    \caption{Illustration of Axioms \ref{axm:mass}-\ref{axm:crossview} through toy examples.}
    \label{fig:axioms}
\end{figure}

Given the above representation, our next aim is to define a means of quantifying the suspiciousness of an arbitrary multi-view subgraph (MVSG). This problem is important in practice, given its implications on manual review prioritization, and practitioner decision-making on enforcement and actioning against various discovered behaviors.  Our basic intuition is that a subset of nodes (groups of entities) which are highly similar to each other (have considerable edge weights between them in multiple views) are suspicious.  \emph{But how can we compare across different MVSGs with varying edge weights, sizes and views in a principled way?  For example, which is more suspicious: 5 organizations with the same IP address and URL, or 10 organizations with the same postal code and creation date?}  The answer is not obvious; this motivates us to formally define criteria that any MVSG scoring metric \emph{should} obey, and define principled ways of quantifying how suspiciousness increases or decreases in these contexts. 

To do so, we first propose several preliminaries which aid us in formalizing these criteria. Informally, we consider an MVSG $\set{X}$ of $n \leq N$ nodes and $k \leq K$ views as a subset of nodes, views or both from $\set{G}$; we denote this relationship compactly as $\set{X} \subseteq \set{G}$, and sometimes abuse notation (when clear) to refer to the associated node set as $\set{X}$.  We introduce similar indexing notation as in the MVG case, such that $\set{X}_i$ and $\set{X}_{\vec{k}}$ refer to a specific subgraph view, and a subset of subgraph views ($\vec{k}$ is a $K$-length vector indicator) respectively.  We define the \emph{mass} of $\set{X}_i$ as $c_i = \sum_{(a,b) \in \set{X}^2} \, w_i^{(a,b)}$, which represents the total sum of edge weights for all edges between nodes in $\set{X}$.  We define the \emph{volume} of $\set{X}_i$ as $v =$ \emph{n choose 2} $=n(n-1)/2$, which denotes the possible number of edges between $n$ nodes. Note that the volume of $\set{X}_i$ is invariant to the view chosen and is only dependent on $n$ (thus, we drop the subscript).  Finally, we define the \emph{density} of $\set{X}_i$ as the ratio between its mass and volume, or $\rho_i = {c_i}/{v_i}$.  We define analogs for mass, volume and density of the associated MVG $\set{G}_i$ with $C_i$, $V$ and $\Rho_i$ respectively.  In general, upper-case variables denote properties of $\set{G}$, while lower-case letters denote properties of $\set{X}$. Given these terms, which are summarized in Table 1, we propose the following axioms below which should be satisfied by an MVSG scoring metric. Note that $ f(\cdot) $ and $ \hat{f}(\cdot) $ represent the desired MVSG scoring metric paramerized by subgraph mass and density, respectively.

\begin{axiom}[Mass] \label{axm:mass}
Given two subgraphs $\set{X}, \set{X}' \subseteq \set{G}$ with the same volume, and same mass in all except one view s.t. $c_i > c_i'$,  $\set{X}$ is more suspicious. Formally,
\begin{displaymath}
c_i > c_i' \Rightarrow f(n,\vec{c},N,\vec{C}) > f(n,\vec{c}',N,\vec{C})
\end{displaymath} 
\end{axiom}
Intuitively, more mass in a view indicates increased attribute similarities between entities, which is more suspicious. For example, it is more suspicious for a group of users to all share the same profile picture, than to all have different profile pictures. 

\begin{axiom}[Size] \label{axm:size}
Given two subgraphs $\set{X}, \set{X}' \subseteq \set{G}$ with same densities $\vec{\rho}$, but different volume s.t. $v > v'$,  $\set{X}$ is more suspicious. Formally,
\begin{displaymath}
v > v' \Rightarrow \hat{f}(n,\vec{\rho},N,\vec{P}) > \hat{f}(n',\vec{\rho},N,\vec{P})
\end{displaymath} 
\end{axiom}
Intuitively, larger groups which share attributes are more suspicious than smaller ones, controlling for density of mass. For example, 100 users sharing the same IP address is more suspicious than 10 users doing the same.

\begin{axiom}[Contrast] \label{axm:contrast}
Given two subgraphs $\set{X} \subseteq \set{G}$, $\set{X}' \subseteq \set{G}'$ with same masses $\vec{c}$ and size $v$, s.t. $\set{G}$ and $\set{G}'$ have the same density in all except one view s.t. $\Rho_i < \Rho_i'$,  $\set{X}$ is more suspicious.  Formally,
\begin{displaymath}
P_i < P_i' \Rightarrow \hat{f}(n,\vec{\rho},N,\vec{P}) > \hat{f}(n,\vec{\rho},N,\vec{P}')
\end{displaymath} 
\end{axiom}
Intuitively, a group with fixed attribute synchrony is more suspicious when background similarities between attributes are rare. For example, 100 users using the same IP address is generally more rare (lower $\Rho_i$) than 100 users all from the same country (higher $\Rho_i'$).

\begin{axiom}[Concentration] \label{axm:concentration}
Given two subgraphs $\set{X}, \set{X}' \subseteq \set{G}$ with same masses $\vec{c}$ but different volume s.t. $v < v'$, $\set{X}$ is more suspicious.  Formally,
\begin{displaymath}
v < v' \Rightarrow f(n,\vec{c},N,\vec{C}) > f(n',\vec{c},N,\vec{C})
\end{displaymath} 
\end{axiom}
Intuitively, a smaller group of entities sharing the same number of similarities is more suspicious than a larger group doing the same. For example, finding 10 instances (edges) of IP sharing between a group of 10 users is more suspicious than finding the same in a group of 100 users.

\begin{axiom}[Cross-view Distribution] \label{axm:crossview}
Given two subgraphs $\set{X}, \set{X}' \subseteq \set{G}$ with same volume $v$ and same mass in all except two views $i, j$ with densities $\Rho_i < \Rho_j$ s.t. $\set{X}$ has $c_i = M, c_j = m$ and $\set{X}'$ has $c_i = m, c_j=M$ and $M>m$, $\set{X}$ is more suspicious. Formally,
\begin{gather*}
\Rho_i < \Rho_j \; \wedge \; c_i > c_i' \; \wedge \; c_j < c_j' \; \wedge \; c_i+c_j=c_i'+c_j' \Rightarrow \\
f(n,\vec{c},N,\vec{C}) > f(n,\vec{c}',N,\vec{C}) 
\end{gather*}
\end{axiom}

Intuitively, a fixed mass is more suspicious when distributed towards a view with higher edge rarity. For example, given 100 users, it is more suspicious for 100 pairs to share IP addresses (low $\Rho_i$) and 10 pairs to share the same country (high $\Rho_j$), than vice versa.  This axiom builds from Axiom \ref{axm:contrast}.


Figure \ref{fig:axioms} illustrates these axioms via toy examples. Informally, these axioms assert that when other subgraph attributes are held constant, suspiciousness constitutes: higher mass (Axiom \ref{axm:mass}), larger volume with fixed density (Axiom \ref{axm:size}), higher sparsity in overall graph (Axiom \ref{axm:contrast}), higher density (Axiom \ref{axm:concentration}), and more mass distributed in sparser views (Axiom \ref{axm:crossview}). These axioms serve to formalize desiderata, drawing from intuitions stemming from prior works; however, prior works \cite{jiang2016inferring, hooi2016fraudar} do not consider multi-view cases, and as we show later are unable to satisfy these axioms.  Notably, such metrics produce unexpected results when scoring MVSGs, and thus lead to  misaligned expectations in resulting rankings. We propose the following problem:


\begin{problem}[MVSG Suspiciousness Scoring]
\label{prob:scoring}
\textbf{Given} an MVG $\set{G}$, \textbf{define} an MVSG scoring metric $f: \set{S} \rightarrow \mathbb{R}$ over the set $\set{S}$ of candidate MVSGs which satisfies Axioms \ref{axm:mass}-\ref{axm:crossview}.
\end{problem}
We discuss details of our solution in Section  \ref{sec:metric}.

\subsection{Mining Suspicious Groups}
\label{sec:probform-msg}

Given an MVSG scoring metric $f$, our next goal is to automatically extract MVSGs which score highly with respect to $f$.  This is a challenging problem, as  computing $f$ for each possible MVSG in $\set{G}$ is intractable; there are $2^N - 1$ non-empty candidate node subsets, and $2^K - 1$ non-empty view subsets over which to consider them.  Clearly, we must resort to intelligent heuristics to mine highly suspicious MVSGs while avoiding an enumeration strategy.  We make several considerations in approaching this task.  Firstly, since suspiciousness is clearly related to shared attribute behaviors, we propose exploiting our data representation to identify candidate ``seeds'' of nodes/entities which are promising to evaluate in terms of $f$.  

Moreover, we focus on mining MVSGs $\set{X}$ (WLOG) given a constraint on the numbers of views such that $|\vec{k}|_1 = z$.  This is for a few reasons: Firstly, it is not straightforward to compare the suspiciousness of two MVSGs with varying numbers of views, as these are defined on probability spaces with different numbers of variables.  Secondly, we consider that (a) entities may not exhibit suspicious behaviors in all $K$ views/attributes simultaneously, but rather only a subset, and (b) in evaluation, practitioners can only interpretably parse a small number of relationship types between a group at once; thus, in practice we choose a constraint $z \leq K$ is generally small and can be suitably chosen and adapted according to empirical interpretability.  In effect, this simplifies our problem as we only consider evaluating and mining MVSGs defined over $z$ views, so that we consider \emph{(K choose z)} total view subsets, rather than $2^K - 1$.  We propose the following problem:

\begin{problem}[Mining Suspicious MVSGs]
\label{prob:mining}
\textbf{Given} an MVG $\set{G}$ on $K$ views, a suspiciousness metric $f: \set{S} \rightarrow \mathbb{R}$ over the set $\set{S}$ of candidate MVSGs, and an interpretability constraint $z \leq K$, \textbf{find} the highest scoring MVSGs $\set{X} \subseteq \set{G}$ (WLOG) for which $|\vec{k}|_1 = z$.
\end{problem}

We detail our approach and implementation strategy for solving this problem in Section \ref{sec:method}.
\section{Proposed Suspiciousness Metric}
\label{sec:metric}

We propose an MVSG scoring metric based on an underlying data model for $\set{G}$ in which undirected edges between the $N$ nodes are distributed i.i.d within each of the $K$ views.  For single graph views, this model is widely known as the Erd\"{o}s-R\'{e}nyi (ER) model \cite{newman2002random}.  The ER model is a standard ``null'' model in graph theory, as it provides a framework to reason about a graph with pure at-chance node interactions, and is a reasonable assumption given no prior knowledge about how a govem graph is structured.  In our scenario, we consider two extensions to such a model unexplored by prior work: (a) we consider multiple graph views, and (b) we consider weighted cell values instead of binary ones.  The first facet is necessary to support multi-attribute or multi-view settings, in which behaviors in one view may be very different than another (i.e. shared IP addresses may be much rarer than shared postal codes).  The second facet is necessary to support continuous edge weights $w_i^{(a,b)} \in \mathbb{R}^+$ capable of richly describing arbitrarily complex notions of similarity between multiple entities (i.e. incorporating both number of shared properties, as well as their rarities).  To handle these extensions, we propose the Multi-View Erd\"{o}s-R\'{e}nyi-Exponential model (\mvere):

\begin{defn}[\mvere Model] A multi-view graph $\set{G}$ generated by the \mvere model is defined such that $w_i^{(a,b)} \thicksim Exp(\lambda_i)$ for all edges $a \leftrightarrow b$ in $\set{G}_i$.
\end{defn}

The \mvere model is a natural fit in our setting for several reasons.  Firstly, the Exponential distribution is continuous and defined on support $\mathbb{R}^+$, which is intuitive as similarity is generally non-negative.  Secondly, it has mode 0, which is intuitive given that sharing behaviors are sparse (most entities should not share properties), and the likelihood of observing high-similarity drops rapidly. 

Given that there are  $V = {N(N-1)}/{2}$ edges (including 0-weight edges) in each view, we can derive the closed-form  MLE simply as $\lambda_i={N(N-1)}/{(2C_i)} = {V}/{C_i} = \Rho_i^{-1}$.  From this, we can write the distribution of single-view MVSG mass as follows:

\begin{lemma}[MVSG subgraph mass]
The mass $M_i$ of a \mvere-distributed subgraph of $Exp(\lambda_i)$ follows $M_i \thicksim Gamma(v, \Rho_i^{-1})$
\end{lemma}
\begin{proof}
This follows from convolution of the Exponential distribution; given $M_i \sim Exp(\lambda)$, $\sum_{i=1}^{g} M_i \thicksim Gamma(g, \lambda)$.
\end{proof}


This enables us to define the \emph{suspiciousness} of a given MVSG $\set{X}$ across multiple views in terms of the likelihood of observing some quantity of mass in those views.  Our metric is defined as

\begin{defn}[MVSG Scoring Metric $f$]
The suspiciousness, $f$, of an MVSG $\set{X}$ with $M_i \thicksim Gamma(v, \Rho_i^{-1})$ and volume $v$ is the negative log-likelihood of its mass $\vec{c}$ under the \mvere model:
\begin{displaymath} 
f\left( n,\vec{c},N,\vec{C} \right) = -\log \left( \prod_{i = 1}^{K} Pr \left( M_i = c_i \right) \right)
\end{displaymath}
\end{defn}

We can write this metric in longer form as follows:
\begin{align*} 
& f\left( n,\vec{c},N,\vec{C} \right) = -\log \left( \prod_{i = 1}^{K} Pr \left( M_i = c_i \right) \right) \\
&= \sum_{i=1}^{K}   -v\log\left(\frac{V}{C_i}\right) + \log\Gamma(v) - (v-1)\log c_i - \frac{V c_i}{C_i} \\
&= \sum_{i=1}^{K}   v \log \frac{C_i}{V} + v \log v - v - \log v  - v \log c_i + \log c_i + \frac{V c_i}{C_i}
\end{align*} 

The last line is due to $\log\Gamma(v) = \log v! - \log v$, after which we use Stirling's approximation to simplify $\log v! \approx v \log v - v$. It is sometimes convenient to write suspiciousness in terms of densities $\vec{\rho}, \vec{\Rho}$; thus, we also introduce a so-parameterized variant $\hat{f}$ where we use $\rho_i = c_i / v$ and $\Rho_i = C_i / v$ and simplify as
\begin{displaymath}
\hat{f}\left( n,\vec{\rho},N,\vec{\Rho} \right) =
\sum_{i=1}^{K} v \log \left( \Rho_i \right) - v \log \rho_i - v + \log \rho_i + v\frac{\rho_i}{\Rho_i}
\end{displaymath}

The intuition for this metric is that high suspiciousness is indicated by low probability of observing certain mass. Since we are interested MVSGs with unlikely \emph{high} density (indicating synchrony between entities), we consider only cases where $\rho_i > \Rho_i$ for all views, to avoid focusing on excessively sparse MVSGs.

\begin{lemma}[Adherence to Axioms]
\label{lem:adherence}
Our proposed suspiciousness metric $f$ (and $\hat{f}$) satisfies each of the MVSG scoring Axioms \ref{axm:mass}-\ref{axm:crossview}.
\end{lemma}
\begin{proof}
We give the full proofs in Section \ref{sec:repro}.
\end{proof}

We note that while the model and metric above considers independent views, we could easily consider adapting the assumed model to account for arbitrary factorizations of the joint distribution, given this knowledge.  However, choosing a more complex  factorization is a problem-specific inference task, and also raises issues with the curse of dimensionality which is broader than the scope of our work. We employ the independence assumption due to its limited estimation challenges under sparsity, generality, and demonstrated usefulness in a wealth of prior works despite their simplicity. 

\subsection{Issues with Alternative Metrics}

\begin{table}[t!]
\setlength{\tabcolsep}{3pt}
\footnotesize
\caption{\label{tbl:related_comparison} Comparison with alternative metrics.}
\centering
\begin{tabular}{l | lllll || l}
Axiom Adherence  & \rotatebox[origin=c]{70}{\mass \cite{lee2010survey}} & \rotatebox[origin=c]{70}{\avgdeg \cite{charikar2000greedy}} & \rotatebox[origin=c]{70}{\dens \cite{lee2010survey}} & \rotatebox[origin=c]{70}{\singval \cite{prakash2010eigenspokes}} & \rotatebox[origin=c]{70}{\cssusp \cite{jiang2016spotting}} & \rotatebox[origin=c]{70}{\;\method} \\ 
\hline
Mass  & \tick & \tick & \tick & \tick & \textcolor{darky}{\textbf{?}} \tablefootnote{\cssusp is limited to discrete edge counts and cannot handle continuous mass settings.}& \tick \\
Size   & \tick & \tick& \cross& \tick & \textcolor{darky}{\textbf{?}} $^1$& \tick\\
Contrast &\cross & \cross&\cross & \cross&\textcolor{darky}{\textbf{?}} $^1$ & \tick\\
Concentration  & \cross& \tick&\tick & \tick & \textcolor{darky}{\textbf{?}} $^1$ & \tick\\
Cross-View Distr.  & \cross&\cross &\cross &\cross &\cross & \tick\\
\hline
\end{tabular}
\end{table}

One may ask, \emph{why not use previously established metrics of suspiciousness?} We next show that these metrics produce results which violate one or more proposed Axioms, and are thus unsuitable for our problem. We compare their performance on the toy settings from Figure \ref{fig:axioms}.  Each subgraph pair illustrates one of Axioms \ref{axm:mass}-\ref{axm:crossview}, and the shaded figures indicates higher intuitive suspiciousness. We consider 5 alternative metrics: mass (\mass) and density (\dens) \cite{lee2010survey}, average degree (\avgdeg) \cite{charikar2000greedy}, singular value (\singval) \cite{prakash2010eigenspokes} and \cssusp (metric from \cite{jiang2016spotting}). 

\textbf{Overview of Alternative Metrics.} \label{sec:overview-alternative} Prior work has suggested \mass, \avgdeg and \dens as suspiciousness metrics for single graph views \cite{shin2016mzoom, beutel2013copycatch,hooi2016fraudar}. We extend these to multi-view cases by construing an \emph{aggregated} view with edge weights summed across the $K$ views.
\cite{jiang2016spotting} proposes \cssusp for suspiciousness in discrete, multi-modal tensor data; we can apply this by construing an MVSG $\set{X}$ as a 3-mode tensor of $n \times n \times K$.  In short, other metrics are agnostic to differences across views, hence they all \emph{violate Axiom \ref{axm:crossview}} (Cross-View Distribution).  Below, we discuss the specifics of each alternative metric, and their limitations with respect to Axioms \ref{axm:mass}-\ref{axm:concentration}. 

\mass: Mass is defined as $\sum_{i=1}^{K} c_i$, or the sum over all edge weights and views. Table \ref{tbl:related_comparison} shows that it violates Axioms \ref{axm:size} (Size) and \ref{axm:contrast} (Contrast) by not considering subgraph size $v$ or graph density $\Rho_i$.

\avgdeg: Average degree is defined as $\sum_{i=1}^{K} c_i/{n}$, or average \mass per node. It does not consider graph density $P_i$ and thus violates Axiom \ref{axm:contrast} (Contrast).

\dens: Density is defined as $\sum_{i=1}^{K} c_i / v$, or average \mass per edge, over $v$ edges. It trivially violates Axiom \ref{axm:size} (Size) by not considering the ratio of subgraph density and size. It also violates Axiom \ref{axm:contrast} (Contrast) by not considering graph density $P_i$.

\singval: Singular values are factor ``weights'' associated with the singular value decomposition $\mathbf{A} = \mathbf{U \Sigma V}^T$; here, we consider the leading singular value $\Sigma_{\scriptstyle 0,0}$ over the view-aggregated $\mathbf{A}$.  \cite{shah2014spotting} shows that for i.i.d. cells in $\mathbf{A}$, $\Sigma_{\scriptstyle 0,0} = \sqrt{n\rho_i}$, though this does not hold generally. Under this assumption, the metric violates Axiom \ref{axm:contrast} (Contrast), by not considering graph density $P_i$.

\cssusp: \cite{jiang2016spotting} defines block suspiciousness as $-\log(Pr(M = \sum_{i=1}^{K} c_i))$, where $M$ is subtensor mass under assumption that cells are discrete, Poisson draws.  However, this constrains adherence to Axioms \ref{axm:mass}-\ref{axm:concentration} (Mass, Size, Contrast and Concentration) only for discrete settings, and is unusable for continuous edge weights/cell values.  This limitation is notable, as later shown in Sections \ref{sec:implementation} and \ref{sec:eval}.   

\section{Proposed Algorithm: {\protect \methodT}}
\label{sec:method}

Given the metric defined in the previous section, we next aim to efficiently mine highly suspicious groups, as proposed in  Problem \ref{prob:mining}. At a high-level, our approach is to start with a small MVSG over a few nodes and views, and expand via a greedy, alternating maximization approach (between nodes and views), evaluating $f$ until a local optimum is reached. Our main considerations are twofold: \emph{How can we scalably expand a given seed until convergence?} and \emph{How can we select good seeds in the first place?} We next address these questions.

Our goal is to find MVSGs $\set{X} \subseteq \set{G}$ (WLOG) which score highest on $f$, and also meet the interpretability constraint $|\vec{k}|_1 = z$.
As mentioned in Section \ref{sec:probform-msg}, full enumeration is combinatorial and computationally intractable in large data.  
Thus, we resort to a greedy approach which allows scalable convergence to locally optimal MVSGs.  Our approach, \method, is outlined in Algorithm \ref{alg:main}.


In short, the algorithm begins by seeding an MVSG defined over a few nodes and $z$ views according to some notion of suitability, and then utilizes an alternating maximization approach to improve the seed: the node set is kept fixed while the view set is updated, and subsequently the view set is kept fixed while the node set is updated.  The updation steps only occur when $f$ increases, and since suspiciousness is bounded for any given $\mathcal{X}$ (i.e. there are a finite number of possible choices for nodes and views), we ensure convergence to a local optimum.  Next, we discuss updating and seeding choices, where we cover the referenced updation and seeding methods.


\textbf{Updating choices.} In order to find a highly suspicious group of entities, we aim to optimize the view set and node set selection via the $\updateviews$ and $\updatenodes$ methods. $\updateviews$ can be written concisely as  $\operatorname*{argmax}_{\scriptstyle \vec{k}} f\left( n,\vec{c},N,\vec{C} \right)$, subject to $|\vec{k}| = z$. This is straightforward given our metric; we independently choose the top-$z$ most suspicious views, given the fixed node set from the prior iteration. For $\updatenodes$, we limit our search space to adding or removing a single node in the MVSG, which is dramatically more tractable than the $2^{N} - 1$ possible node set choices over $\set{G}$.  We write $\updatenodes$ concisely as $\operatorname*{argmax}_{\scriptstyle \set{X}'} f\left( n,\vec{c},N,\vec{C} \right)$, subject to $|\mathcal{X}' \setminus \mathcal{X}| + |\mathcal{X} \setminus \mathcal{X}'| \leq 1$, meaning that each update changes the node set by, at most, a single entity (one more or one less).

\textbf{Seeding choices.}  Clearly, update quality relies on reasonable seeding strategy which is able to find candidate suspicious MVSGs and also explore $\set{G}$.  The $\seedviews$ method can be achieved in multiple ways; ultimately, the goal is to chooses $z$ initial views such that the seeds expand to a diverse set suspicious MVSGs. Given the desire for diversity, a reasonable approach is to sample $z$ views uniformly as done in prior work \cite{jiang2014inferring, shin2016mzoom}.  However, a downside with random sampling is that it does not respect our intuition regarding the non-uniform value of entity similarity across views.  For example, consider that a view of \emph{country} similarity has only 195 unique values, whereas a view of \emph{IP Address} similarity has $2^{32}$ unique values; naturally, it is much more common for overlap in the former than the latter, despite the latter having a higher signal-to-noise ratio.  Thus, in practice, we aim to sample views in a weighted fashion, favoring those in which overlap occurs less frequently. We considered using the inverse of view densities $\vec{\rho}$ as weights, but we observe that density is highly sensitive to outliers from skewed value frequencies. We instead use the inverse of the $q^{th}$ frequency percentiles across views as more robust estimates of their overlap propensity ($q \geq 95$ works well in practice). The effect is that lower signal-to-noise ratio views such as \emph{country} are more rarely sampled.   

Defining $\seednodes$ is more challenging.  As Section \ref{sec:metric} mentions, we target overly dense MVSGs $\set{X}$ (WLOG) for which $\rho_i > \Rho_i$ for all views.  The key challenge is identifying seeds which satisfy this constraint, and thus offer promise in expanding to more suspicious MVSGs.  Again, one could consider randomly sampling node sets and discarding unsatisfactory ones, but given that there are $z$ constraints (one per view), the probability of satisfaction rapidly decreases as $z$ increases (Section \ref{sec:eval} elaborates).  To this end, we propose a carefully designed, greedy seeding technique called \seedMethod (see Algorithm \ref{alg:greedyseed}) which enables us to quickly discover good candidates.  Our approach exploits that that satisfactory seeds occur when entities share more similarities, and strategically constructs a node set across views which share properties with each other.  Essentially, we initialize a candidate seed with two nodes that share a similarity in a (random) one of the chosen views, and try to incrementally add other nodes connected to an existing seed node in views where the constraint is yet unsatisfied.  If unable to do so after a number of attempts, we start fresh.  The process is stochastic, and thus enables high seed diversity and diverse suspicious MVSG discovery.



\begin{algorithm}[!t]
\scriptsize
\caption{{\sc \method}\label{alg:main}} 
\begin{algorithmic}[1]
\Require MVG $\set{G}$ ($N$ nodes, $K$ views, $\vec{C}$ masses), constraint $z \leq K$
\State $\vec{k} \gets \seedviews(\set{G}, z)$ \Comment{\text{\color{blue} choose $z$ views}}
\State $\set{X}_{\vec{k}} \gets \seednodes(\set{G}, \vec{k})$ \Comment{\text{\color{blue} $n$ nodes, $\vec{c}$ masses}}
\State $S \gets f \left( n, \vec{c}, N, \vec{C} \right)$ \Comment{\text{\color{blue} compute suspiciousness metric}}
\Do  \Comment{\text{\color{blue} alternating optimization}}
\State $S' \gets S$
\State $\vec{k} \gets \updateviews(\set{G}, \set{X})$ \Comment{\text{\color{blue} revise view set}}
\State $\set{X}_{\scriptstyle \vec{k}} \gets \updatenodes(\set{G}, \set{X}_{\scriptstyle \vec{k}})$ \Comment{\text{\color{blue} revise node set}}
\State $S \gets f \left(n, \vec{c}, N, \vec{C} \right)$
\doWhile{$S > S'$} \Comment{\text{\color{blue} repeat until $S$ converges}}
\State \Return $\left( \set{X}_{\vec{k}}, \, S \right)$ 
\end{algorithmic}
\end{algorithm} 

\begin{algorithm}[!t]
\scriptsize
\caption{{\sc GreedySeed}\label{alg:greedyseed}} 
\begin{algorithmic}[1]
\Require MVG $\set{G}$ ($N$ nodes, $K$ views, $\Rho$ dens.), views $\vec{k}$

\State define $\shuffle(S)$: return S in random order
\State define $\choice(S, r)$: return $r$ random elements from S

\State $\set{V} \gets \left\{ i \mid \mathbbm{1}(k_i) \right\}$ \Comment{\text{\color{blue} chosen view set}}
\State $\fun{H}_{i}^{ve}(a) \gets a \Rightarrow \fun{A}_i^{-1}(a) \; \forall a \in \set{A}, i \in \set{V} $  \Comment{\text{\color{blue} value-to-entity hashmap}}
\State $\fun{H}_{i}^{ev}(a) \gets a \Rightarrow \set{A}_{i}(a)  \; \forall a \in \set{G}, \; i \in \set{V} $  \Comment{\text{\color{blue} entity-to-value hashmap}}
\State  $i \gets \choice(\set{V}, 1)$ \Comment{\text{\color{blue} choose a view}} \label{alg:gs-retry}
\State  $a \gets \choice(\left\{ a \mid |\fun{H}_{i}^{ve}(a) \;| \; \geq 2 \right\}, 1)$ \Comment{\text{\color{blue} choose a shared value}}
\State  $\set{X} \gets \choice(\fun{H}_{i}^{ve}(a), \; 2)$  \Comment{\text{\color{blue} initialize seed with similar entities}}

\For {view $i \in \shuffle(\set{V})$}
    \State  $t \gets 0$ \Comment{\text{\color{blue} keep track of attempts}}
    \While { ($\rho_i < \Rho_i$ and t < 20) } \Comment{\text{\color{blue} try to satisfy constraint}}
            \State $e_1 \in \shuffle(\set{X}, 1) $ \Comment{\text{\color{blue} choose entity already in $\set{X}$}}
            \State $a \gets \choice(\fun{H}_{i}^{ev}(e_1), 1)$ \Comment{\text{\color{blue} choose a shared value}}
            \State $e_2 \gets \choice(\fun{H}_{i}^{ve}(a), 1)$ \Comment{\text{\color{blue} choose a similar entity}}
            \State $\set{X} \gets \set{X} \cup e_2$ \Comment{\text{\color{blue} grow seed}}
            \State $t \gets t + 1$ 
    \EndWhile
    \If {($\rho_i < \Rho_i$)}
        \State \Goto{alg:gs-retry} \Comment{\text{\color{blue} start fresh if constraint not yet met}}
    \EndIf
\EndFor
\State \Return $\set{X}_{\vec{k}}$
\end{algorithmic}
\end{algorithm}

\subsection{Implementation}
\label{sec:implementation}
We describe several considerations in practical implementation, which improve result relevance and computational efficiency.

\textbf{Result quality.} 
The quality of resulting MVSGs depends on the degree to which they accurately reflect truly suspicious synchronous behaviors. Not all synchronicity is equally suspicious, such as for attributes from free-form user text input. Consider an attribute ``File Name'' for user uploads. An edge between two users which share an uploaded file named ``Cheap-gucci-ad.jpg'' is, intuitively, more suspiciouss than if the file was named ``Test'' (which is a common placeholder that unrelated users are likely to use). To avoid considering the latter case, which is a type of \emph{spurious} synchronicity, we use techniques from natural language processing to carefully weight similarities in the MVG construction. Firstly, we enable value blacklisting for highly common stop-words (e.g. \emph{Test}).  Overlap on a stopworded value produces 0 mass (no penalty).  Next, we weight edges for other value overlaps according to the common TF-IDF NLP technique, which in this case is best characterized as a value's \emph{ inverse entity frequency} (IEF) \cite{aggarwal2012mining}. We define the IEF for value $v$ in view $i$ as: $ief(v_i) = \left({N} / \log\left(1 + |\set{A}_i^{-1}\left(v_i \right)|\right)\right)^2$, which significantly discounts common values. We let $w_i^{(a,b)} = \sum_{v_i \in \set{A}(a) \cap \set{A}(b)} ief(v_i)$, so that edge weight between two nodes (entities) depends on both number and rarity of shared values.  

\textbf{Computational efficiency.}  Next, we discuss several optimizations to improve the speed of suspicious MVSG discovery.  Firstly, we observe that \method is trivially parallelizable. In our implementation, we are able to run thousands of seed generation and expansion processes simultaneously by running in a multi-threaded setting, and aggregating a ranked list afterwards.  Another notable performance optimization involves the computation of view masses $\vec{c}$ in $\updatenodes$, which is the slowest part of \method.  A na\"{i}ve approach to measure synchronicity given $n$ nodes is to quadratically evaluate pairwise similarity, which is highly inefficient. We instead observe that it is possible to compute $c_i$ by cheaply maintaining the number of value frequencies in a view-specific hashmap $\fun{J}_i$, such that $\fun{J}_{i}(v) = |\left\{e \in \set{X} \mid v \in \fun{A}_{i}(e) \right\}|$. Specifically, $\fun{J}_{i}\left(v\right)$ indicates that ${\fun{J}_{i}\left(v\right)}^2 - \fun{J}_{i}\left(v\right)$ similarities exist in the subgraph view on the value $v$, and since each of them contribute $ief(v)$ weight, we can write the total mass as $c_i = \sum_{v_i \in \set{A}_i} ief(v_i)\left( \fun{J}_{i}\left(v_i\right) \right)^2 - \fun{J}_{i}\left(v_i\right)$.  This approach makes it possible to calculate view mass in linear time with respect to the number of subgraph nodes, which drastically improves efficiency. Furthermore, the space requirements are, for each view, a hashmap of value frequencies as well as a table of unique attribute values for each entity, a compressed representation compared to tensor-based approaches~\cite{prakash2010eigenspokes}. The time and space complexity requirements are more formally described in Section VI-C.  
\section{Evaluation}
\label{sec:eval}

Our experiments aim to answer the following questions.

\begin{compactitem}
\item \textbf{Q1. Detection performance:} How effective is \method in detecting suspicious behaviors in real and simulated settings? 
How does it perform in comparison to prior works?

\item \textbf{Q2. Efficiency:} How do \seedMethod and \method scale theoretically and empirically on large datasets?
\end{compactitem}

\subsection{Datasets}
We use both real and simulated settings in evaluation.

\textbf{Snapchat advertiser platform.} We use a dataset of advertiser organization accounts on Snapchat, a leading social media platform. Our data consists of $N=230K$ organizations created on Snapchat from 2016-2019. Each organization is associated with several single-value (e.g. contact email) and multi-value attributes (e.g names of ad campaigns). All in all, we use $K=12$ attributes deigned most relevant to suspicious behavior detection, made available to us by Snapchat's Ad Review team, whom we partner with for domain expert analysis and investigation:

\begin{compactitem}
\item \textbf{Account details (6):}  Organization name, e-mails, contact name, zip-code, login IP addresses, browser user agents. 
\item \textbf{Advertisement content (6):} Ad campaign name (e.g. \emph{Superbowl Website}), media asset hashes, campaign headlines (e.g. \emph{90\% off glasses!}), brand name (e.g. \emph{GAP}), iOS/Android app identifiers, and external URLs. 
\end{compactitem}

Based on information from the domain experts, we pruned 1.7K organizations from the original data, primarily including advertisement agencies and known affiliate networks which can have high levels of synchrony (often marketing for the same subsets of companies), and limited our focus to the remaining $228K$ organizations.

\textbf{Simulated settings.} We additionally considered several simulated attack settings, based on realistic attacker assumptions.  Our intuition is that attack nodes will have higher propensity to share attribute values than normal ones, and may target varying attributes and have varying sophistication. Our simulation parameters include $N$ (entity count) and $K$ (total attribute types), $\vec{u}$ (length-$K$, cardinalities of attribute value spaces), $n,k$  (entities and attributes per attack), $c$ (number of attacks), $\lambda$ (value count per normal entity), and $\tau$ (attack temperature, s.t. attackers choose from a restricted attribute space with cardinalities $\tau^{-1}\vec{u}$.  Together, these parameters can express a wide variety of attack types.  Our specific attack procedure is:
\begin{compactenum}
\item Initialize $N$ normal entities. For each attr. $i$, draw $Poisson(\lambda)$ specific values uniformly over $[1,u_i]$.
\item Conduct an attack, by randomly sampling $n$ entities and $k$ attributes.  For each attr. $i$, draw $Poisson(2\lambda)$ specific attr. values uniformly over $[1, \tau^{-1}u_i]$.  Repeat $c$ times.
\end{compactenum}
Unless otherwise mentioned, for each scenario we fix parameters as $N=500$ nodes, $K=10$ views, $u_i = 50i$ attr. cardinality, $n=50$ nodes per attack, $k=3$ views per attack, $\lambda=5$ mean values drawn per node and attribute, and $\tau=10$ temperature.

\subsection{Detection performance}

\begin{figure}
    \centering
    \includegraphics[width=0.98\linewidth]{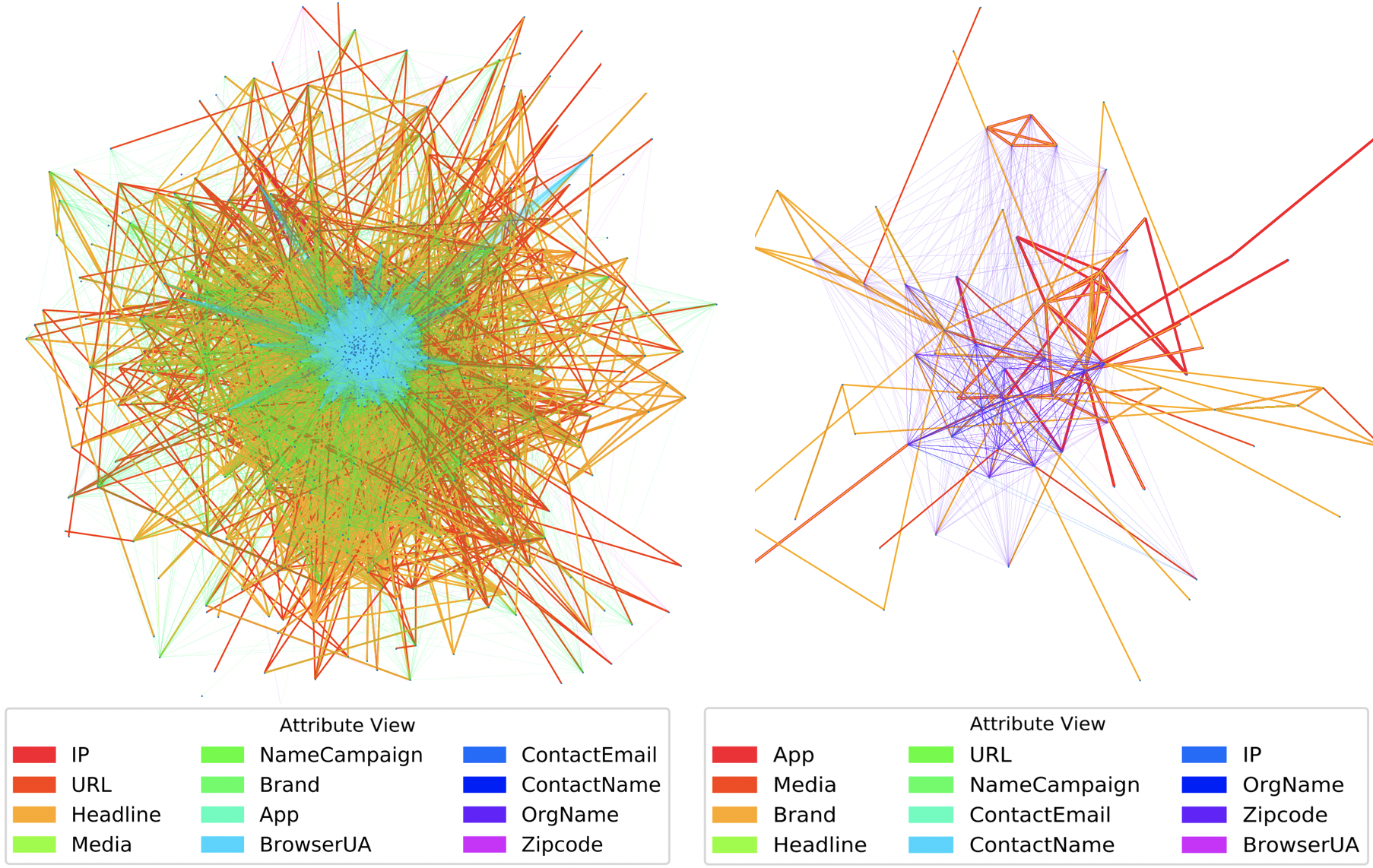}
    \caption{\method detects both blatant (left) and more stealthy (right) fraudsters on Snapchat's ad platform. Attribute Views in each Legend are sorted by suspiciousness (Red = Highest, Purple = Lowest). 
    }
    \label{fig:eval_viz}
\end{figure}

\begin{figure*}[!t]
\centering
\subcaptionbox{High attack sync.}
{
\includegraphics[width=0.18\textwidth]{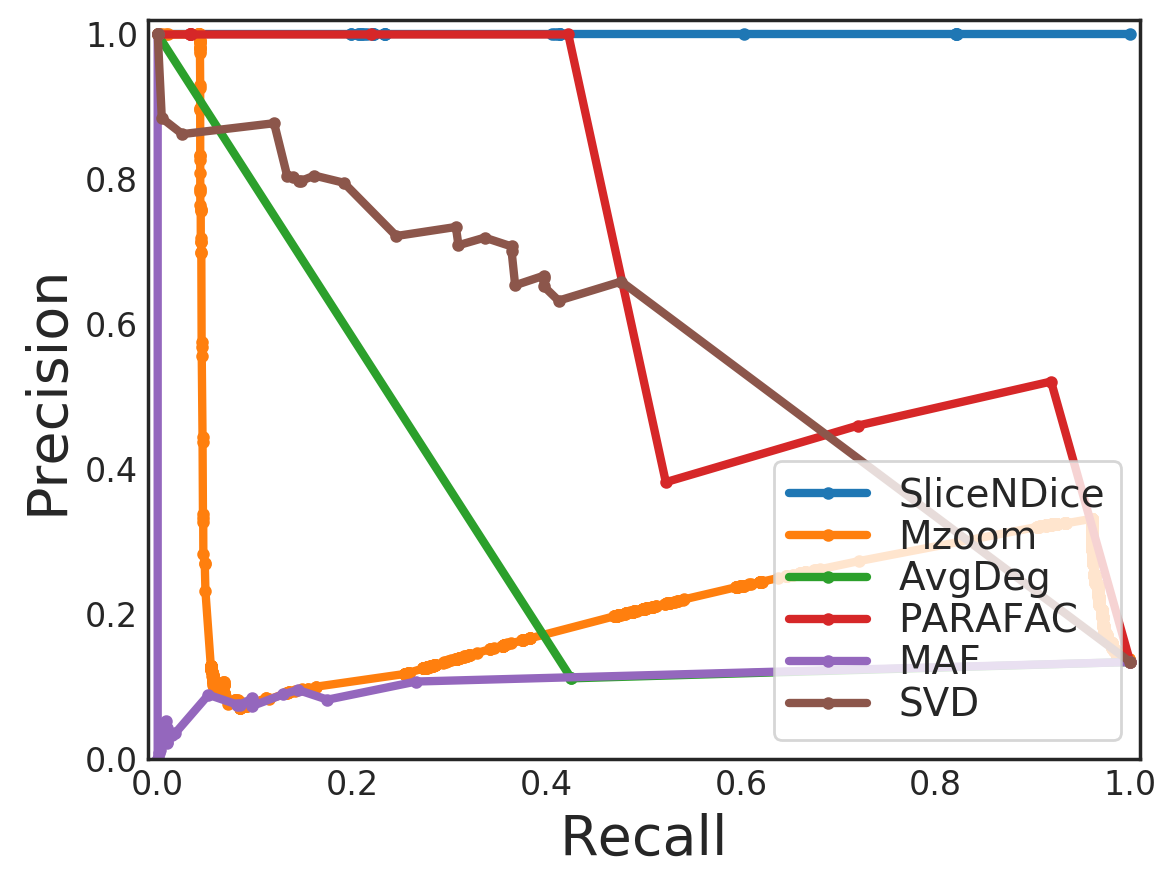}
}
\subcaptionbox{Low attack sync.}
{
\includegraphics[width=0.18\textwidth]{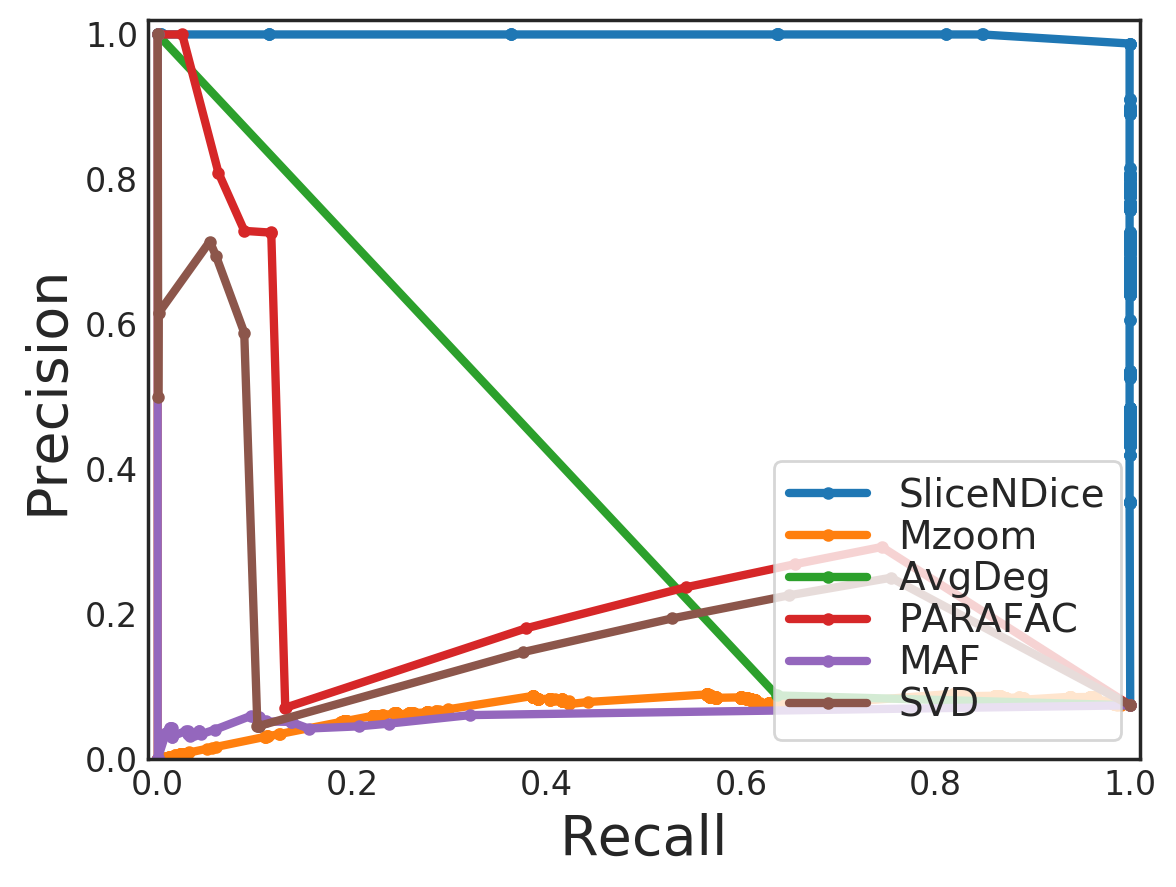}
}
\subcaptionbox{High-signal attacks}
{
\includegraphics[width=0.18\textwidth]{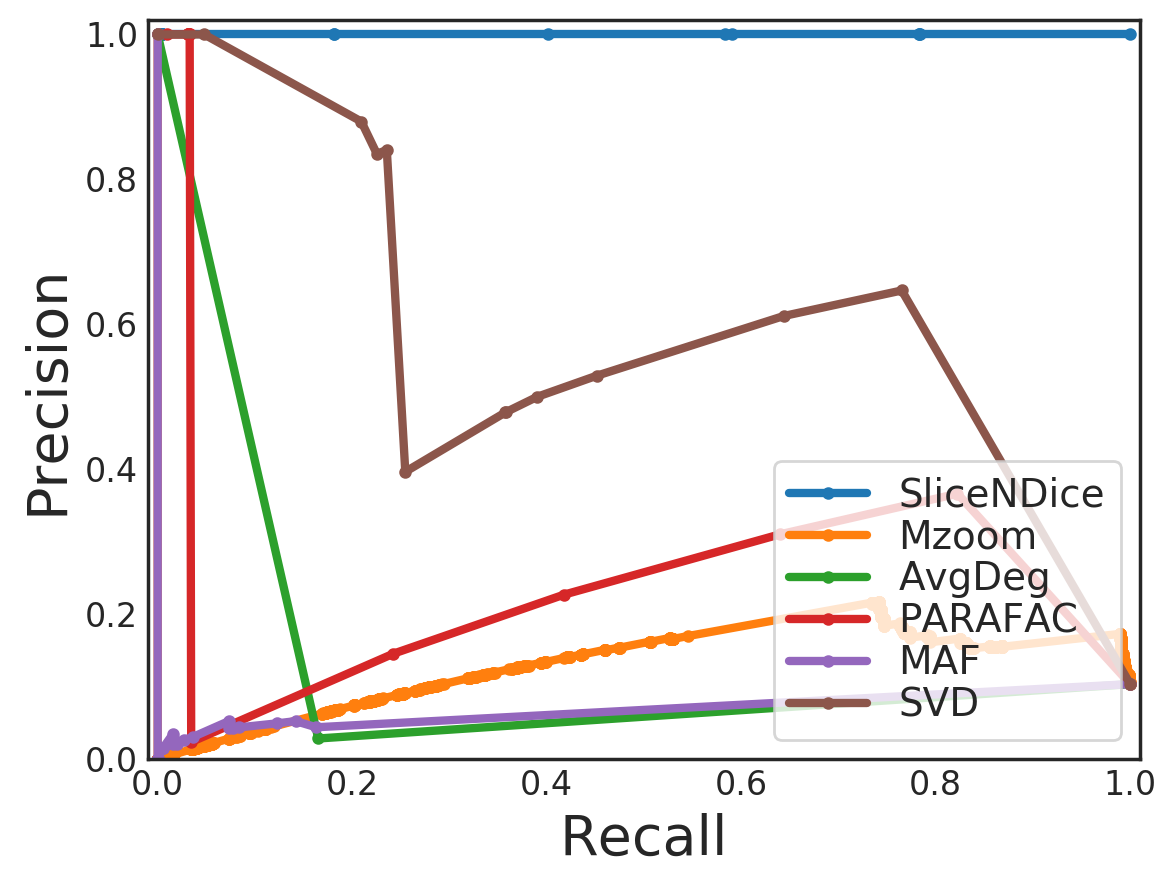}
}
\subcaptionbox{Low-signal attacks}
{
\includegraphics[width=0.18\textwidth]{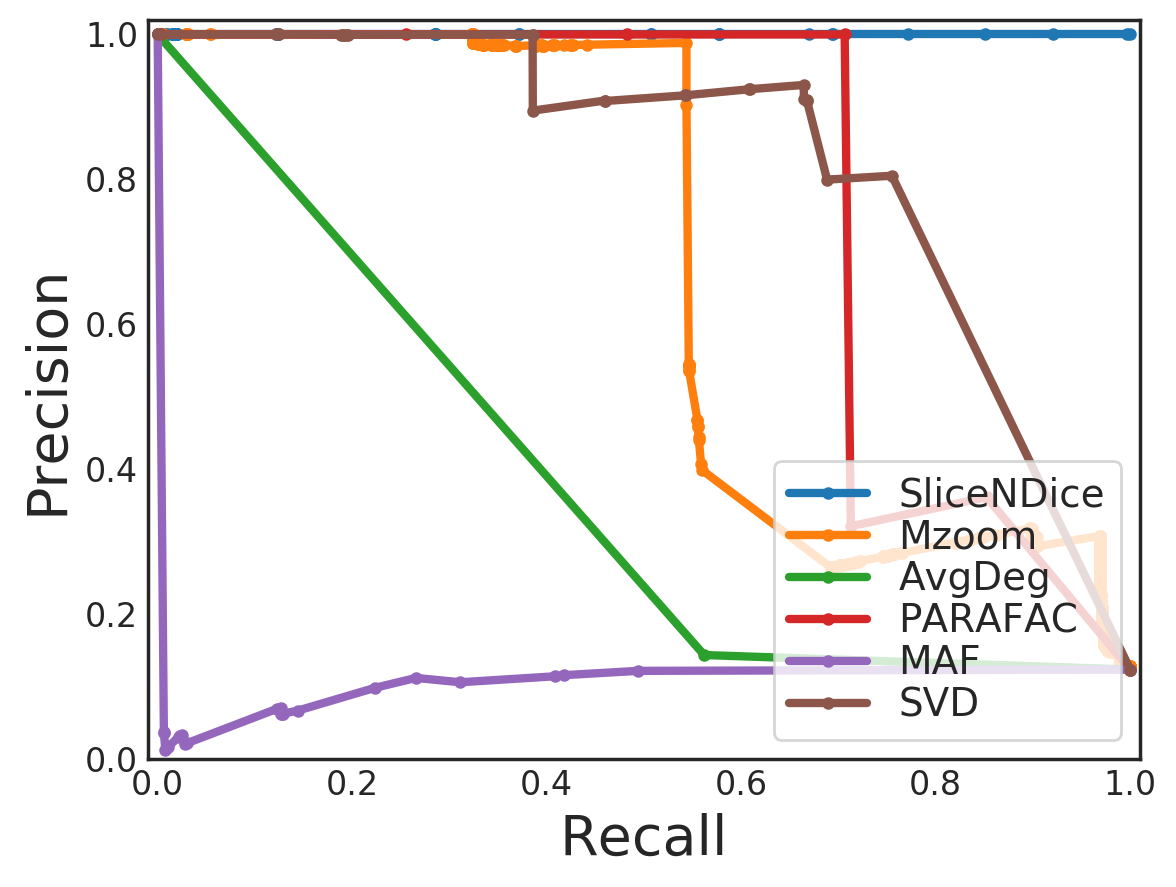}
}
\subcaptionbox{High-dim attacks}
{
\includegraphics[width=0.18\textwidth]{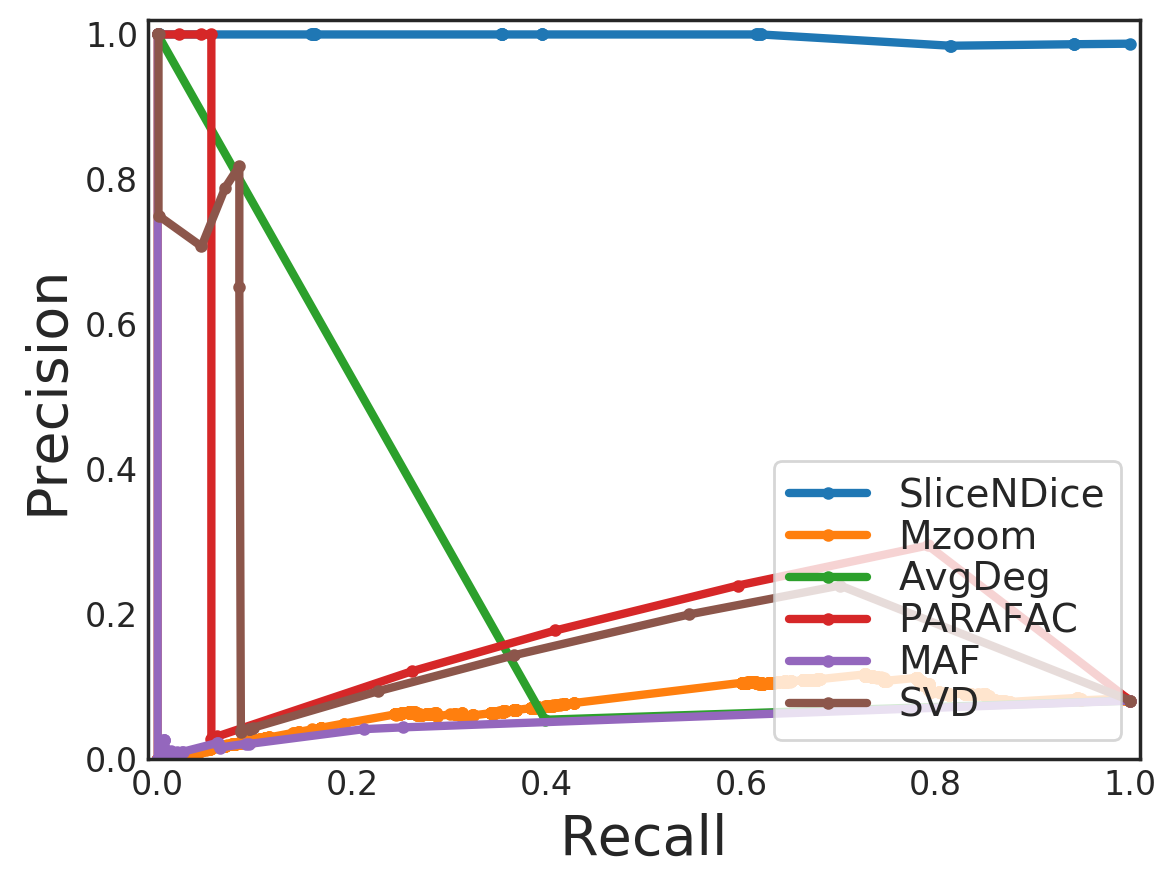}
}
\caption{\method (blue, near the top) consistently achieves extremely high precision/recall on various, realistic attack settings, despite changing attack complexity.  Other methods fail due to their inability to respect differences between attributes/views and overemphasis on global density. \label{fig:synth-eval-attacks}}
\end{figure*}

We discuss performance in both settings.

\textbf{Snapchat advertiser platform.} 
We deployed \method on Google Cloud compute engine instances with high vCPU count to enable effective parallelization across many seeds over 2 days.  During this time, we yielded a total of $6,050$ suspected entity group behaviors within Snapchat data. We fixed $z = 3$ for this run, based on two reasonS: (a) input from our Ad Review partners, who noted that too many apparent features hindered the review process via information overload, and (b) $z=3$ balances the two extremes of low $z$ preventing discovery of more distributed, stealthier fraud which can only be uncovered by overlaying multiple graph views, and high $z$ hurting interpretability and increasing difficulty to satisfy density constraints. Finally, after mining and ranking many MVSGs, we aggregate and filter results for expert review  by pruning ``redundant'' MVSGs which covering the same set of nodes; we use a Jaccard similarity threshold $\eta=0.05$ to determine overlap.
We note that the number of distinct blocks was fairly robust to $\eta$; for example, a more liberal threshold of $\eta = .25$ yielded $636$ distinct groups. We chose $\eta$ conservatively in order to minimize redundancy in results shared with the Ad Review team. We evaluated our methods both quantitatively and qualitatively in coordination with them. We were not able to compare with other suspicious group mining methods which rely on block decomposition or search~\cite{papalexakis2013more, prakash2010eigenspokes, charikar2000greedy, shin2016mzoom, jiang2016spotting} as these are highly dense matrices/tensors which are too large to manifest for this scale.


\textit{Quantitative Evaluation:} We sort the suspected sub-graphs in descending order and submitted those in the top 50 groups to the Ad Review team for in-depth, manual validation of the constituent $2736$ organizations.  We provided the domain experts with 3 assets to aid evaluation: (a) network visualizations like those in Figure \ref{fig:eval_viz}, (b) mappings between the highest frequency occurring attribute values for each attribute, and all organization entities associated with those attributes, and (c) entity-level indices listing all instances of attribute synchrony and associated values, for each organization involved per group. While the exact review process and signals used are masked for security reasons, the review process generally consisted of evaluating individual advertiser accounts in each cluster, and determining their suspiciousness based on previous failed payments, previously submitted spammy or fraudulent-looking ads, similarity with previously discovered advertiser abuse vectors, and other proprietary signals.  From surveying these organizations, reviewers found that an overwhelming $2435$ were connected to fraudulent behaviors that violated Snapchat's advertiser platform Terms of Service, resulting in an organization-level precision of $89\%$.  The organizations spanned diverse behaviors, including individuals who created multiple accounts to making multiple (similar) accounts to generate impressions before defaulting early on their spending budget (avoiding payment), those selling counterfeit goods or running e-commerce scams, fraudulent surveys with falsely promised rewards, and more.  Intuitively, the diversity of these behaviors supports the flexibility of our problem formulation; despite \method's agnosticism to the multiple types of abuse, it was able to uncover them automatically using a shared paradigm.


\textit{Qualitative Evaluation:} We inspect two instances of advertiser fraud, shown in Figure \ref{fig:eval_viz}, with help from the Ad Review team.  Although \method selected groups based on suspiciousness according to only the top ($z=3$) views, we show the composites of similarities for all 12 attributes to illustrate the differences between the two attacks.  On the left, we show the first case of \textit{blatant} fraud across $500$ organizations, which are connected across all of the considered attributes to varying degrees.  Many of these organizations were accessed using a common subset of IP addresses, with the cluster near the center all sharing browser user agents. Upon further inspection, we find that these accounts are engaging in e-commerce fraud, and link to a series of URLs that follow the pattern \emph{contact-usXX.[redacted].com}, where $XX$ ranges from 01-27. Multiple accounts link to these common URLs and share common advertisement headlines, which combined with shared login activities ranks the group very highly according to our metric. Our second case study is of a smaller ring of $70$ organizations, which could be considered \textit{stealthy} fraud. These organizations appear to market a line of Tarot card and Horoscope applications.  We noticed attempts taken by the fraudsters to cloak this ring from detection: no app identifier appears more than 4 times across the dataset, but most of the organizations have identifiers associating to similar app variants.  This discovery illustrates \method's ability to ``string together'' shared properties across multiple attribute views to discover otherwise hard-to-discern behaviors.  

\textbf{Simulated settings.} We consider detection performance on simulations matching several realistic scenarios.
\begin{compactenum}
\item \textbf{High attacker synchrony:} Default settings; attacks sample attributes from $1/10$th of associated attribute spaces, making them much denser than the norm.
\item \textbf{Low attacker synchrony:} $\tau=2$. Attribute spaces are restricted to $1/2$ and thus much harder to detect.
\item \textbf{High-signal attribute attacks:} Attack views are sampled with weights $\propto \vec{u}$ (more likely to land in sparse views).
\item \textbf{Low-signal attribute attacks:} Attack views are sampled with weights $\propto 1/\vec{u}$ (more likely to land in dense views).  
\item \textbf{Attacks in high dimensionality:} $K = 30$.  Attacks are highly distributed in $3\times$ higher dimensionality.
\end{compactenum}
Our detection task is to classify each attribute overlap as suspicious or non-suspicious; thus, we aim to label each nonzero entry (``behavior'') in the resulting $N \times N \times K$ tensor.  We evaluate \method's precision/recall performance along with several state-of-the-art group detection approaches in this task.  For each method, we penalize each behavior in a discovered block using the block score given by that method, and sum results over multiple detected blocks.  The intuition is that a good detector will penalize behaviors associated with attack entities and attributes more highly.  We compare against PARAFAC decomposition \cite{mao2014malspot}, MultiAspectForensics (MAF) \cite{maruhashi2011multiaspectforensics}, Mzoom \cite{shin2016mzoom}, SVD \cite{prakash2010eigenspokes}, and AvgDeg \cite{charikar2000greedy}.  For SVD and AvgDeg which only operate on single graph views, we use the aggregated view adjacency matrix, whereas for others we use the adjacency tensor.  \method utilizes the compressed representation discussed in Section \ref{sec:implementation}.  For SVD, we use singular value ($\singval$) as the block score.  For PARAFAC and MAF, we use the block norm which is effectively the higher order $\singval$.  AvgDeg uses the average subgraph degree (\avgdeg), and Mzoom measures mass suspiciousness under the Poisson assumption (\cssusp). Section \ref{sec:repro} gives further details regarding baseline comparison.

Figure \ref{fig:synth-eval-attacks} shows precision/recall curves for all 5 attack scenarios; note that \method significantly outperforms competitors in all cases, often maintaining over 90\% precision while making limited false positives.  Alternatives quickly drop to low precision for substantial recall, due to their inability to both (a) find and score groups while accounting for differences across attributes (Axiom \ref{axm:crossview}), and (b) correctly discern the suspicious from non-suspicious views, even when the right subgraph is discovered.  In practice, this allow attributes with low signal-to-noise ratio and higher natural density to overpower suspicious attacks which occur in less dense views.

\subsection{Efficiency}

We consider efficiency of both \method, as well as our seeding algorithm, \seedMethod.  The time complexity of \seedviews is trivially $\set{O}(z)$, as it involves only choosing $z$ random views.  We can write the complexity of any suitable method for \seednodes loosely as $\set{O}(z \bar t)$, given $z$ views and $\bar t$ iterations to satisfy each density constraint $\rho_i > \Rho_i$ successively.  However, in practice, this notion of $\bar t$ is ill-defined and can adversely affect performance.  We explore practical performance in average time to find suitable seeds which satisfy constraints, for both random seeding and our \seedMethod: Figure \ref{fig:scal-seeding} shows that our \seedMethod finds seeds $100-1000\times$ faster than random seeding on real data; note the log scale.  Random seeding struggles significantly in meeting constraints as the number of views increases, further widening the performance gap between the methods.  Each iteration of \updatenodes is $\set{O}(NK\bar A)$ given $N$ nodes, $K$ views, and $\bar A$ values per attribute.  \updateviews is simply $\set{O}(z \log z)$, as \updatenodes already updates  subgraph masses $\vec{c}$ across all views.  The runtime in practice is dominated by \updatenodes; assuming $\bar T$ iterations per MVSG, the overall \method time complexity is $\set{O}(NK\bar A \bar T)$.  Figures \ref{fig:scal-nodes}-\ref{fig:scal-iter} show that \method scales linearly with respect to both entities (over fixed iterations) and iterations (over fixed entities).  Moreover, the overall space complexity is $\mathcal{O}(K \bar A)$, due to the compact attribute-oriented data representation described in Section \ref{sec:implementation}.  Note that alternative group mining methods which rely on block decomposition or search were infeasible to run on real data, due to the sheer attribute sharing density in the tensor representation, which grows quadratically with each shared  value.

\begin{figure}[!t]
\subcaptionbox{Time to seed \label{fig:scal-seeding}}
{
\includegraphics[width=0.29\linewidth]{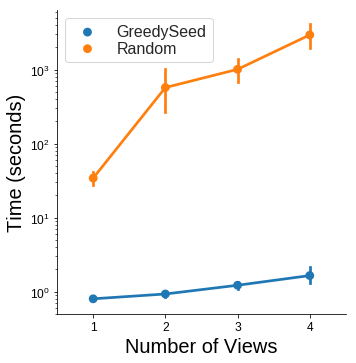}
}
\subcaptionbox{Time vs. \# entities \label{fig:scal-nodes}}
{
\includegraphics[width=0.29\linewidth]{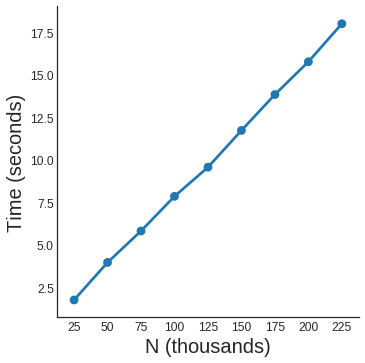}
}
\subcaptionbox{Time vs. \linebreak\# iterations \label{fig:scal-iter}}
{
\includegraphics[width=0.29\linewidth]{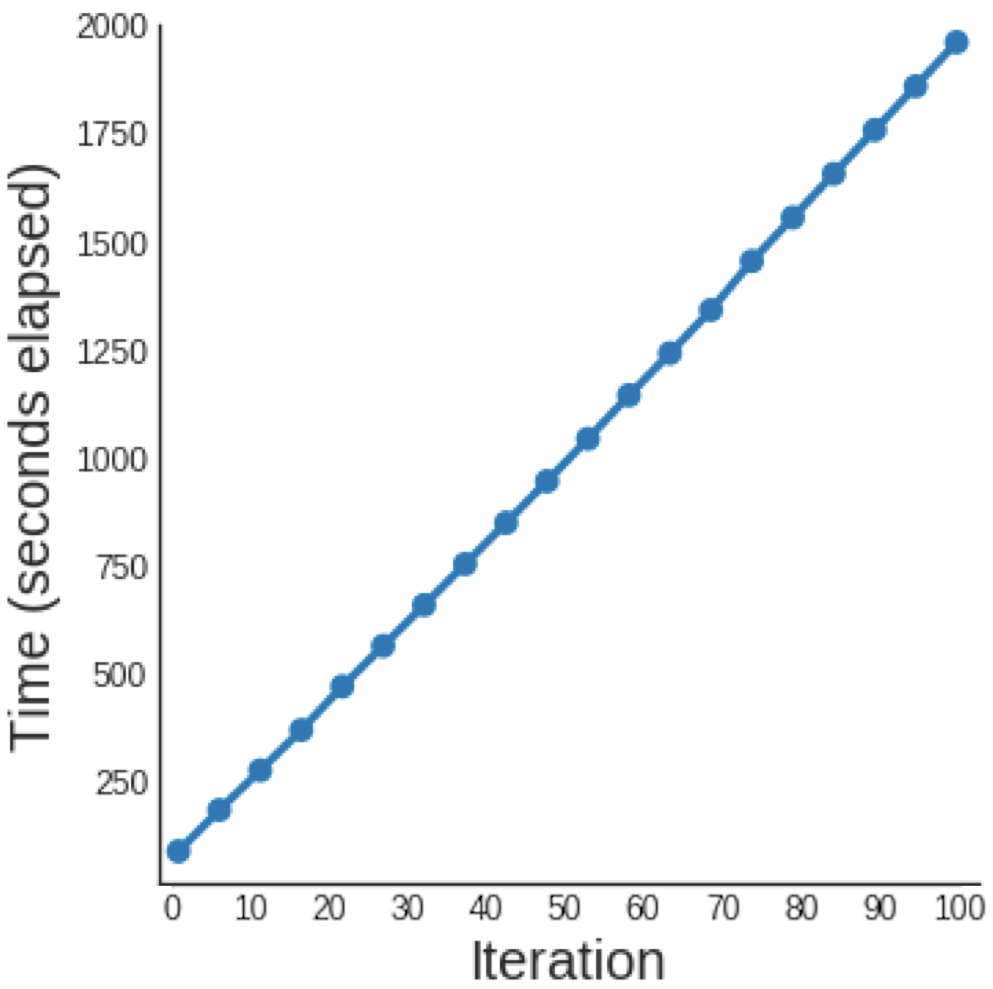}
}
\caption{Our \seedMethod finds suitable seeds $100-1000\times$ faster than random seeding (a).  Moreover, \method scales linearly in number of entities (b) and iterations (c).}
\end{figure}

\section{Conclusion}

In this work, we tackled the problem of scoring and discovering suspicious behavior in multi-attribute entity data.  Our work makes several notable contributions.  Firstly, we construe this data as a multi-view graph, and formulate this task in terms of  mining \emph{suspiciously dense multi-view subgraphs} (MVSGs).  We next propose and formalize intuitive desiderata (Axioms \ref{axm:mass}-\ref{axm:crossview}) that MVSG scoring metrics should obey to match human intuition, and designed a novel suspiciousness metric based on the proposed \mvere model which satisfies these metrics, unlike alternatives.  Next, we proposed the \method algorithm, which enables scalable ranking and discovery of MVSGs suspicious according to our metric, and discussed practical implementation details which help result relevance and computational efficiency.  Finally, we demonstrated strong empirical results, including experiments on real data from the Snapchat advertiser platform where we achieved 89\% precision over 2.7K organizations and uncovered numerous fraudulent advertiser rings, consistently high precision/recall (over 97\%) and outperformance of several state-of-the-art group mining algorithms, and linear scalability.
\bibliographystyle{IEEEtran}
\bibliography{paper.bib}

\clearpage
\section{Reproducibility}
\label{sec:repro}

\subsection{Satisfaction of Axioms}



Below, we show that our suspiciousness metric $f$ (and $\hat{f}$) satisfies Axioms \ref{axm:mass}-\ref{axm:crossview}, and posited in Lemma \ref{lem:adherence}.  In each case, we consider how $f$ or $\hat{f}$ changes as individual properties vary; since they are simply reparameterizations of one another, we suffice it to prove adherence to each axiom using the more convenient parameterization.  We reproduce the axioms with each proof below for reader convenience.

\begin{axiomnonumber}[Mass]
Given two subgraphs $\set{X}, \set{X}' \subseteq \set{G}$ with the same volume, and same mass in all except one view s.t. $c_i > c_i'$,  $\set{X}$ is more suspicious. Formally,
\begin{displaymath}
c_i > c_i' \Rightarrow f(n,\vec{c},N,\vec{C}) > f(n,\vec{c}',N,\vec{C})
\end{displaymath} 
\end{axiomnonumber}
\begin{proof}[Proof of Axiom \ref{axm:mass} (Mass)]
\begin{displaymath} 
\frac{\partial{f_i}}{\partial c_i} = \frac{V}{C_i}-\frac{v-1}{c_i}
\end{displaymath} 
Because we are operating under the constraint $P_i < p_i$, then $P_i^{-1} > p_i^{-1}$ which implies that $\frac{V}{C_i} > \frac{v}{c_i} > \frac{v-1}{c_i} $. Therefore, $\frac{\partial{f_i}}{\partial c_i} > 0$ and so as mass increases, suspiciousness increases. 
\end{proof}

\begin{axiomnonumber}[Size]
Given two subgraphs $\set{X}, \set{X}' \subseteq \set{G}$ with same densities $\vec{p}$, but different volume s.t. $v > v'$,  $\set{X}$ is more suspicious. Formally,
\begin{displaymath}
v > v' \Rightarrow \hat{f}(n,\vec{p},N,\vec{P}) > \hat{f}(n',\vec{p},N,\vec{P})
\end{displaymath} 
\end{axiomnonumber}
\begin{proof}[Proof of Axiom \ref{axm:size} (Size)]
The derivative of $\hat{f}_i$ with respect to subgraph volume is:
\begin{displaymath} 
\frac{\partial\hat{f}_i}{\partial v} =
\log  \Rho_i  - \log \rho_i  - 1 + \frac{\rho_i}{\Rho_i} = \frac{\rho_i}{\Rho_i} - \log \frac{\rho_i}{\Rho_i}  - 1
\end{displaymath} 
This implies that $\frac{\partial\hat{f}_i}{\partial v} > 0$, because  $x - \log x > 1$ always holds when $x > 1$ (which holds given $\rho_i > \rho_i$).  Therefore, suspiciousness increases as volume increases.
\end{proof}

\begin{axiomnonumber}[Contrast] 
Given two subgraphs $\set{X} \subseteq \set{G}$, $\set{X}' \subseteq \set{G}'$ with same masses $\vec{c}$ and size $v$, s.t. $\set{G}$ and $\set{G}'$ have the same density in all except one view s.t. $\Rho_i < \Rho_i'$,  $\set{X}$ is more suspicious.  Formally,
\begin{displaymath}
P_i < P_i' \Rightarrow \hat{f}(n,\vec{p},N,\vec{P}) > \hat{f}(n,\vec{p},N,\vec{P}')
\end{displaymath} 
\end{axiomnonumber}
\begin{proof}[Proof of Axiom \ref{axm:contrast} (Contrast)]
The derivative of $\hat{f}_i$ suspiciousness with respect to density $\Rho_i$ is:
\begin{displaymath}
\frac{\partial \hat{f}_i}{\partial \Rho_i} =
v \log \Rho_i + v \frac{\rho_i}{\Rho_i} = \frac{v}{\Rho_i} - v \frac{\rho_i}{{\Rho_i}^2} =
\frac{v}{\Rho_i} \left( 1 - \frac{\rho_i}{\Rho_i} \right) 
\end{displaymath}
This implies that $\frac{\partial \hat{f}_i}{\partial \Rho_i} < 0$, because $v/\Rho_i > 0$ and $\rho_i > \Rho_i$.  Thus, as graph density increases, suspiciousness decreases.  Alternatively, as sparsity increases, suspiciousness increases.
\end{proof}

\begin{axiomnonumber}[Concentration] 
Given two subgraphs $\set{X}, \set{X}' \subseteq \set{G}$ with same masses $\vec{c}$ but different volume s.t. $v < v'$, $\set{X}$ is more suspicious.  Formally,
\begin{displaymath}
v < v' \Rightarrow f(n,\vec{c},N,\vec{C}) > f(n',\vec{c},N,\vec{C})
\end{displaymath} 
\end{axiomnonumber}
\begin{proof}[Proof of Axiom \ref{axm:concentration} (Concentration)]
The derivative of view $i$'s contribution suspiciousness (parameterized by mass) and w.r.t the volume is:
\begin{align*} 
\frac{\partial{f_i}}{\partial v} &= 
\log\frac{C_i}{V} + \log v + 1 - 1 - \frac{1}{v} - \log c_i \\
&=\log P_i + \log \rho_i^{-1} - \frac{1}{v} = \log\frac{P_i}{\rho_i} - \frac{1}{v} = -\log\frac{\rho_i}{P_i} - \frac{1}{v} 
\end{align*}
$-\log\frac{\rho_i}{P_i} <= -1$ because $\frac{\rho_i}{P_i} > 1$, so therefore $\frac{\partial{f_i}}{\partial v} <  0$. Therefore, for a fixed sub-graph mass $c_i$, suspiciousness decreases as volume increases. 
\end{proof}

\begin{axiomnonumber}[Cross-view Distribution] 
Given two subgraphs $\set{X}, \set{X}' \subseteq \set{G}$ with same volume $v$ and same mass in all except two views $i, j$ with densities $\Rho_i < \Rho_j$ s.t. $\set{X}$ has $c_i = M, c_j = m$ and $\set{X}'$ has $c_i = m, c_j=M$ and $M>m$, $\set{X}$ is more suspicious. Formally,
\begin{gather*}
\Rho_i < \Rho_j \; \wedge \; c_i > c_i' \; \wedge \; c_j < c_j' \; \wedge \; c_i+c_j=c_i'+c_j' \Rightarrow \\ f(n,\vec{c},N,\vec{C}) > f(n,\vec{c}',N,\vec{C}) 
\end{gather*}
\end{axiomnonumber}

\begin{proof}[Proof of Axiom \ref{axm:crossview} (Cross-View Distribution)]

Assume that view $i$ is sparser than view $j$ ($\Rho_i < \Rho_k$), and we are considering a sub-graph which has identical mass in both views ($c_i = c_j$). Adding mass to the sparser view $i$ will increase suspiciousness more than adding the same amount of mass to the denser view $j$ because:

\begin{align*} 
\frac{\partial{f_i}}{\partial v} - \frac{\partial{f_j}}{\partial v} &= \left(\frac{V}{C_i}-\frac{v-1}{c_i}\right) - \left(\frac{V}{C_j}-\frac{v-1}{c_j}\right) \\ &=  P_i^{-1} -  P_j^{-1} > 0
\end{align*} 
\end{proof}

\subsection{Baseline Implementations for Comparison}

We compared against 5 baselines in Section \ref{sec:eval}: PARAFAC \cite{mao2014malspot}, MAF \cite{maruhashi2011multiaspectforensics}, Mzoom \cite{shin2016mzoom}, AvgDeg \cite{charikar2000greedy} and SVD \cite{prakash2010eigenspokes}.  Below, we give background and detail our implementations for these baselines.

\subsubsection{PARAFAC}
PARAFAC \cite{papalexakis2013more} is one of the most common tensor decomposition approaches, and can be seen as the higher-order analog to matrix singular value decomposition. An $F$-rank PARAFAC decomposition aims to approximate a multimodal tensor as a sum of $F$ rank-one factors which, when summed, best reconstruct the tensor according to a Frobenius loss.  In our case, the decomposition produces factor matrices $\mathbf{A}$ of $N \times r$, $\mathbf{B}$ of $N \times r$ and $\mathbf{C}$ of $K \times r$, such that we can write the tensor $\mathbf{\mathscr{T}}$ associated with MVG $\set{G}$ as 
\begin{displaymath}
\mathbf{\mathscr{T}} \approx \sum_{f=1}^F a^f \circ \mathbf b^f \circ c^f
\end{displaymath}
where $a^f$ denotes the $f^{th}$ column vector of $\mathbf{A}$ (analogue for $b^f$ and $c^f$), and $[\vec{a} \circ \vec{b} \circ \vec{c}](i,j,k) = a_i b_j c_k$.  Since PARAFAC gives continuous scores per node in the $\vec{a}$ and $\vec{b}$ vectors for each rank-one factor, we sum them and then use the decision threshold suggested in \cite{shin2016mzoom} ($2.0/\sqrt{N}$) to mark nodes above that threshold as part of the block. We then select the top $k$ views which individually have the highest singular value over the associated submatrix, and penalize the associated entries with the norm of the $f^{th}$ rank-one tensor (closest approximation of \singval in higher dimensions).  We use the Python \texttt{tensorly} library implementation, and $F = 5$ decomposition.

\subsubsection{MAF}
MAF \cite{maruhashi2011multiaspectforensics} also utilizes PARAFAC decomposition, but proposes a different node inclusion method.  Their intuition is to look for the largest ``bands'' of nodes which have similar factor scores, as they are likely clusters.  Since in our case, $\vec{a}$ and $\vec{b}$ both reflect node scores, we sum them to produce the resulting node factor scores.  We then compute a log-spaced histogram over these using 20 bins (as proposed by the authors), and sort the histogram from highest to lowest frequency bins.  We move down this histogram, including nodes in each bin until reaching the 90\% energy criterion or 50\% bin size threshold proposed by the authors.  We mark these nodes as included in the block, and select the top $k$ views with the highest associated submatrix singular values, as for PARAFAC.  We likewise penalize entries in this block using the associated block norm.  We use the Python \texttt{tensorly} library implementation, and specify a $F = 5$ decomposition.

\subsubsection{Mzoom}
Mzoom \cite{shin2016mzoom} proposes a greedy method for dense subtensor mining, which is flexible in handling various block-level suspiciousness metrics.  It has been shown outperform \cite{jiang2016spotting} in terms of discovering blocks which maximize \cssusp metric, and hence we use it over the method proposed in \cite{jiang2016spotting}.  The algorithm works by starting with the original tensor $\mathscr{T}$, and greedily shaves values from modes which maximize the benefit in terms of improving \cssusp.  When a block $\mathscr{T}'$ is found which maximizes \cssusp over the local search, Mzoom prunes it from the overall tensor in order to not repeatedly converge to that block, and repeats the next iteration with $\mathscr{T} - \mathscr{T}'$. As Mzoom does not allow selection of a fixed $k$ views, we only modify their implementation to add the constraint that for any blocks found with $\geq k$ views, we limit output to the first $k$ for fairness.  We penalize entries in each block with that block's \cssusp score.  We use the authors' original implementation, which was written in Java (and available on their website) and specify 500 blocks to be produced (unless the tensor is fully deflated/empty sooner).

\subsubsection{SVD}
SVD \cite{prakash2010eigenspokes}, as discussed in Section \ref{sec:overview-alternative}, is a matrix decomposition method which aims to produce a low-rank optimal reconstruction of $\mathbf{A}$ according to Frobenius norm. In our case, since we aggregate over the $K$ views and produce a resulting $N \times N$ matrix for $\set{G}$, a rank $F$ SVD decomposes the matrix $\mathbf{A} = \mathbf{U \Sigma V}^T$, where $\mathbf{U}, \mathbf{V}$ are $N \times F$, and $\mathbf{Sigma}$ is $F \times F$ and diagonal, containing the singular values.  Loosely, SVD effectively discovers low-rank structures in the form of clusters, such that $\mathbf{U}, \mathbf{V}$ indicate cluster affinity and $\Sigma$ indicates cluster strength or scale.  We take a similar approach as for PARAFAC, in that for each rank $f$, we sum the factor scores $u^f$ and $v^f$ and use a decision threshold of $2.0/\sqrt{1000}$ to mark node membership in the given block.  Then, we rank the individual views according to their respective leading singular values, and choose the top $k$ for inclusion.  We then penalize  entries in the block with the submatrix \singval score.  We use the Python \texttt{scipy} package, and specifically the \emph{svds} method to run sparse SVD, for which we use an $F = 5$ decomposition.

\subsubsection{AvgDeg}
\cite{charikar2000greedy} proposes an algorithm, which we call AvgDeg, for greedily mining dense subgraphs according to the \avgdeg notion of suspiciousness.  The algorithm proposed gives a 2-approximation in terms of returning the maximally dense subgraph, and works by considering a single graph $\set{G}$, and greedily shaving nodes with minimum degree while keeping track of the \avgdeg metric at each iteration.  Upon convergence, the algorithm returns a subgraph which scores highly given \avgdeg.  As for SVD, we consider $\set{G}$ to be aggregated over all $K$ views.  Though the algorithm proposed by the author was initially defined in terms of unweighted graphs, we adapt it to weighted graph setting by shaving nodes greedily with minimum \emph{weighted degree} rather than the adjacent edge count.  Upon discovery of one subgraph $\set{G}'$, we repeat the next iteration with $\set{G} - \set{G}'$.  After finding the nodes for each subgraph, we choose the top $k$ views with the highest individual \avgdeg for inclusion in the block. We request up to 500 blocks, but in practice find that the algorithm converges very quickly because it incorrectly prunes a large number of nodes in earlier iterations.   

\subsubsection{\method}

We use the standard implementation as described in Section \ref{sec:implementation}, evaluating over 500 blocks.  Our implementation is written in Python, and will be made available publicly.

\subsection{Source Code and Datasets}
All source code including calculation of the proposed suspiciousness metric, and our implementation of the \method algorithm is available at \url{http://github.com/hamedn/SliceNDice/}. Our implementation was done using Python 3.7. The algorithm takes as input a CSV, where each row is an entity and each column is an attribute, and returns a ranked list of suspicious groups by row identifier. The code used to generate simulated attacks as discussed in Section \ref{sec:eval} is also included, and allows researchers and practitioners to create their own simulated attack datasets by modifying simulation parameters: $N$ (entity count), $K$ (total attribute types), $\vec{u}$ (length-$K$, cardinalities of attribute value spaces), $n,k$  (entities and attributes per attack), $c$ (number of attacks), $\lambda$ (value count per normal entity), and $\tau$ (attack temperature, s.t. attackers choose from a restricted attribute space with cardinalities $\tau^{-1}\vec{u}$. We also include benchmarking code used to compare the performance of \method against the aforementioned baselines.  Unfortunately, given that the Snapchat advertiser data contains sensitive PII (personally identifiable information), it is not possible for us to release the  dataset publicly. In fact, to the best of our knowledge, no such real-world social datasets are publicly available in the multi-attribute setting. This is because  although the multi-attribute detection setting is a highly common one in many social platforms, attributed data in these settings is typically PII.

\end{document}